\documentclass[english,runningheads]{llncs}
\usepackage[latin1]{inputenc}
\usepackage[T1]{fontenc}
\usepackage{babel}
\usepackage{graphicx}
\usepackage{amsmath}
\usepackage{amsfonts}
\usepackage{amssymb}
\usepackage{multirow}
\usepackage{multicol}
\usepackage{breakcites}
\usepackage{url}
\usepackage{stmaryrd}
\usepackage{lmodern}
\usepackage{longtable}
\usepackage{rotating}
\usepackage{array}
\usepackage{pdflscape}
\usepackage{soul}
\usepackage{xspace}
\usepackage{enumitem}
\usepackage{hyperref}
\usepackage{color}
\usepackage{wrapfig} 
\usepackage{colortbl} 

\usepackage{proof}
\usepackage{makecell}

\newcommand{\tts}{\tt \small}

\newcommand{\Diff}{${\mathcal R}$}

\newcommand{\If}{\leftarrow}

\newcommand{\eq}{\!=\!}

\newcommand{\vars}{\mathit{vars}}

\newcommand{\repl}{\! \Rightarrow \!}

\newcommand{\Embedded}
{\mbox{\hspace*{.5mm}\raisebox{-1.1mm}{$\sim$}\hspace*{-2.5mm}\raisebox{.2mm}{\large{$\triangleleft$}}\hspace*{.9mm}}}
\newcommand{\Down}{\raisebox{-.5mm}{\rule{0mm}{1mm}}}

\title{Satisfiability of Constrained Horn Clauses on Algebraic Data Types:
\\A Transformation-based Approach\,\thanks{This work has been
partially supported by the National Group of Computing Science (GNCS-INdAM).
This paper is an improved, extended version of a paper that
the authors have presented
at IJCAR 2020~\cite{De&20a} and also at the 35th Italian Conference on
Computational Logic (CILC 2020). }}

\author{
Emanuele De Angelis\,\inst{1}\,\orcidID{0000-0002-7319-8439}\and
Fabio~Fioravanti\,\inst{2}\,\orcidID{0000-0002-1268-7829} \and
Alberto~Pettorossi\,\inst{3,1}\,\orcidID{0000-0001-7858-4032}\and
Maurizio~Proietti\,\inst{1}\,\orcidID{0000-0003-3835-4931}
}

\institute{IASI-CNR, Rome, Italy,
	\email{\{emanuele.deangelis,maurizio.proietti\}@iasi.cnr.it}\\
	\and DEC, University of Chieti-Pescara, Italy,
	\email{fabio.fioravanti@unich.it}\\
	\and \!\!DICII,\,University\,of\,Rome\,Tor\,Vergata,\,Rome,\,Italy,
	\email{\!alberto\!.\!pettorossi@uniroma2\!.\!it}
}

\usepackage{fancyhdr}
\fancypagestyle{firststyle}
{
   \fancyhf{}
   
   \fancyfoot[L]{\small\textit{This work has been SUBMITTED (under consideration)
   to the CILC 2020 special issue of the Journal of Logic and Computation}
   }
}

\begin{document}
\maketitle

\thispagestyle{firststyle}

\begin{abstract}

We address the problem of checking %
the satisfiability of 
Constrained Horn Clauses (CHCs) defined on %
Algebraic Data Types (ADTs), 
such as lists and trees.
We propose a new technique for transforming CHCs defined on ADTs into CHCs 
where the arguments of the predicates have %
only basic types, such as integers and booleans.
Thus, our technique avoids, during satisfiability checking,
the explicit 
use of proof rules based on induction over the ADTs.
The main extension over previous techniques for ADT removal is 
a {new transformation rule,}
called {\it differential replacement}, which allows us to introduce
auxiliary predicates, whose definitions correspond to lemmas that are used
when making inductive proofs.
We present an algorithm that performs the automatic removal of ADTs
by applying the {new}
rule, together with the traditional folding/unfolding rules.
We prove that, under suitable hypotheses, the set of the transformed 
clauses is satisfiable if and only if
so is the set of the original clauses.
By an experimental evaluation, 
we show that the use of the {new} rule 
significantly improves the effectiveness of ADT removal.
We also show
that our approach is competitive
with respect to tools that extend CHC solvers
with the use of inductive rules.

\end{abstract}
\vspace*{-8mm}
\section{Introduction}
\label{sec:Intro}
\label{intro}

{\it Constrained Horn Clauses} (CHCs) constitute a fragment of
the first order predicate calculus, where the Horn clause syntax %
is extended 
by allowing {\it constraints} on specific domains
to occur in clause premises. CHCs have gained popularity as 
a logical formalism %
well suited for automatic verification
of programs~\cite{Bj&15}. %
Indeed, 
many verification problems can be reduced to the satisfiability 
problem for CHCs.

Satisfiability of CHCs is a particular case of
{\it Satisfiability Modulo Theories} (SMT), understood here as the general 
problem of determining the satisfiability of (possibly quantified) 
first order formulas where the interpretation of some
function and predicate symbols is defined in
a given constraint theory (also called {\it background theory})~\cite{BaT18}.
Recent advances in the field have led to the development of a number of
very powerful SMT {\it solvers}
(and, in particular, CHC solvers), which
aim at solving satisfiability problems with respect to a large
variety of constraint theories.
Among the SMT solvers, we would like to mention
CVC4~\cite{CVC4}, MathSAT~\cite{MaS13}, and Z3~\cite{DeB08},
and among solvers with specialized engines for CHCs,
we recall Eldarica~\cite{HoR18}, HSF~\cite{Gr&12}, RAHFT~\cite{Ka&16}, VeriMAP~\cite{De&14b}, and Z3-SPACER~\cite{Ko&13}.

{Even if SMT algorithms for unrestricted first order formulas 
suffer from incompleteness limitations due to general undecidability results,} 
most of the above mentioned tools work well {in practice} {when acting}
on constraint theories, such as 
Booleans, Uninterpreted Function Symbols, Linear Integer or Real Arithmetic, Bit Vectors, and
Arrays.
However, when formulas contain {universally quantified variables}
{ranging over} inductively defined {\it algebraic data types} (ADTs),
such as lists and trees, then 
{the SMT/CHC solvers often show poor results, as they do not incorporate  
induction principles for %
the ADT in use.}

To {mitigate} this difficulty, some SMT/CHC solvers 
have been enhanced by incorporating suitable %
induction principles~\cite{ReK15,Un&17,Ya&19},
similarly to what has been done in automated theorem provers~\cite{Bun01}.
The most creative step which is needed when extending SMT solving 
with induction,
is the generation of the auxiliary
lemmas that are required for proving the main conjecture.

An alternative approach, proposed in the context of CHCs~\cite{De&18a}, 
consists in transforming a given set of clauses 
into a new set: (i) where all ADT terms are removed 
{(without introducing new function symbols)},
and (ii)~whose satisfiability implies the satisfiability 
of the original set of clauses.
This approach has the advantage of separating the concern of dealing
with ADTs (which is considered
at transformation time) from the concern of dealing with
simpler, non-inductive constraint theories (which is 
{considered}
at solving time), thus
avoiding the complex interaction between inductive reasoning and 
constraint solving.
It has been shown~\cite{De&18a} that the transformational approach 
compares well with induction-based solvers %
if lemmas are not needed in the proofs.
However, in some satisfiability problems, if suitable lemmas are 
not provided, the transformation fails to remove the ADT terms.

The main contributions of this paper are as follows.

\noindent \hangindent=6mm 
(1) We extend the transformational approach by proposing a new rule, called 
{\it differential replacement}, based on the introduction of appropriate %
{\it difference predicates},
which play a role similar to that of lemmas in inductive proofs.
We prove that the combined use of the fold/unfold transformation rules~\cite{EtG96}
and the differential replacement rule is {\it sound}, that is,
if the transformed set of clauses is satisfiable, then 
the original set of clauses is satisfiable. 
We also study some sufficient conditions that guarantee that the use of those rules is {{\it sound and complete}} %
in the sense that the transformed set of clauses 
is satisfiable if and only if the original set of clauses is satisfiable.
	
\noindent \hangindent=6mm 
(2) We develop a transformation algorithm that removes ADTs from CHCs 
by applying the fold/unfold and the differential replacement
rules in a fully automatic way. 
	
\noindent \hangindent=6mm 
(3)~{Due to the undecidability of the satisfiability problem for CHCs, in general 
our technique for ADT removal may not terminate. Thus, we evaluate  
its effectiveness from an experimental
point of view and, in particular, we discuss the results obtained by the 
implementation of our technique in a tool, called {\sc AdtRem}}. 
We consider a set of CHC satisfiability problems on ADTs
taken from various benchmarks which are used for evaluating inductive theorem provers.
The experiments show that {\sc AdtRem} is competitive
with respect to Reynolds and Kuncak's tool that augments the CVC4 solver with  inductive reasoning~\cite{ReK15}.

\smallskip

\noindent
The paper is structured as follows. 
In Section~\ref{sec:IntroExample} we briefly present an introductory, 
motivating example.
In Section~\ref{sec:CHCs} we recall some
basic notions about CHCs.
In Section~\ref{sec:TransfRules} we introduce the rules used
in our transformation technique and, in particular, the novel 
differential replacement rule, and we show the soundness of the 
rules we consider. 
In Section~\ref{sec:Strategy} we present a transformation algorithm,
called~\Diff,
that uses the transformation rules for removing ADTs from sets of CHCs.
In Section~\ref{sec:Completeness} we show that, under suitable hypotheses, the 
transformation rules we use are also complete.
In Section~\ref{sec:Experiments} we illustrate the {\sc AdtRem} tool
and we present the experimental results we have obtained.
Finally, in Section~\ref{sec:RelConcl} we discuss the related work and
make a few concluding remarks.

\section{A Motivating Example}
\label{sec:IntroExample}
Let us consider the following functional program {\it Reverse},
which we write using the OCaml syntax~\cite{Le&17}:

\vspace*{-1.5mm}
{\small
\begin{verbatim}
  type list = Nil | Cons of int * list;;
  let rec append l ys = match l with
    | Nil -> ys          | Cons(x,xs) -> Cons(x,(append xs ys));;
  let rec snoc l y = match l with
    | Nil -> Cons(y,Nil) | Cons(x,xs) -> Cons(x,snoc xs y);;
  let rec reverse l = match l with
    | Nil -> Nil         | Cons(x,xs) -> snoc (reverse xs) x;;
  let rec len l = match l with
    | Nil -> 0           | Cons(x,xs) -> 1 + len xs;;
\end{verbatim}
}

\vspace*{-1.5mm}

\noindent
The  functions {\tts append}, {\tts reverse}, and {\tts len}
compute list concatenation, list reversal,
and list length, respectively.
The function {\tts snoc}, given a list~{\tts l} and an element~{\tts y},
returns the list obtained by inserting~{\tts y} at the end of {\tts l}.

Suppose we want to prove the following property concerning those functions:~

\vspace{.5mm}

{\small
$\mathtt{\forall}$
{\tt xs,ys.\ len\,(reverse\,(append\ xs\ ys))\ =\ (len\ xs)\,+\,(len\ ys)}
\hfill Property $(1)$\hspace*{1mm}
}

\vspace{.5mm}

\noindent
This property follows from the facts that: (i)~by appending a
list {\tts xs} of length~{\tts n0} and a list~{\tts ys} of length~{\tts n1},
we get a list of length {\tts n0\,+\,n1}, and (ii)~by reversing a list,
we get a list with the same length. In the program {\it Reverse} we have
assumed that the elements of the lists are integers, but Property~(1)
holds %
independently of the type of the elements of
the lists.
Inductive theorem provers construct a proof of
Property~(1) by induction on the structure of the list {\tts{l}}, by assuming
the knowledge of the following lemma:

\vspace{.5mm}
{\small
	$\mathtt{\forall}$
	{\tt x,xs. len (snoc xs x) = (len xs) + 1}
	\hfill Lemma $(2)$\hspace*{1mm}
}

\vspace{.5mm}

\noindent
which states that, given a list {\tts xs} of length {\tts n} and
an element {\tts x}, {\tts snoc xs x} returns a list of length {\tts n+1}.

The approach we follow in this paper avoids both the explicit
use of induction principles and the %
knowledge of {\it ad hoc} lemmas.
First, we consider the translation of Property~($1$) into a set of
constrained Horn clauses,
where, for every function~$f$ defined by a given functional program,
the atom $f(X,Y)$ is the translation of
`$f(X)$ evaluates to $Y$ in the {\em call-by-value}
semantics'~\cite{De&18a,Un&17}.
(Automated translation techniques have also been proposed for imperative languages 
with functions~\cite{De&17b,Gr&12}.)
We get
set {\it RevCls} of the following clauses\footnote{In the examples,
	we use Prolog syntax for writing clauses, instead of the more verbose
	SMT-LIB syntax. The predicates {\tt =\textbackslash=} (different from), {\tt =} (equal to), {\tt <} (less-than),
	{\tt >=}~(greater-than-or-equal-to)
	denote constraints between integers.
}:

\vspace{-1mm}

{\small
	\begin{verbatim}
	1. false :- N2=\=N0+N1, append(Xs,Ys,Zs), reverse(Zs,Rs),
	            len(Xs,N0), len(Ys,N1), len(Rs,N2).
	2. append([],Ys,Ys).
	3. append([X|Xs],Ys,[X|Zs]) :- append(Xs,Ys,Zs).
	4. reverse([],[]).
	5. reverse([X|Xs],Rs) :- reverse(Xs,Ts), snoc(Ts,X,Rs).
	6. snoc([],Y,[Y]).
	7. snoc([X|Xs],Y,[X|Zs]) :-  snoc(Xs,Y,Zs).
	8. len([],N) :- N=0.
	9. len([X|Xs],N1) :- N1=N0+1, len(Xs,N0).
	\end{verbatim}
}

\vspace{-1mm}

\noindent
{\it RevCls} is satisfiable if and only if
Property~$(1)$ holds.
However, state-of-the-art CHC solvers, such as Eldarica or Z3-SPACER,
fail to prove the satisfiability of the set {\it RevCls} of clauses,
because those solvers do not incorporate any induction principle on lists.

To overcome this difficulty, we may apply the transformational
approach based on the fold/unfold rules~\cite{De&18a}, whose objective is to
transform a given set of clauses into a new set without occurrences of list variables. 
Then, the satisfiability of the derived set of clauses can be checked by
using CHC solvers based on the theory of Linear Integer Arithmetic ({\textit{LIA}}) only.
The soundness of the transformation rules ensures that the satisfiability
of the transformed clauses
implies the satisfiability of the original ones.

In the transformational approach the fold/unfold  rules are applied 
according to a given algorithm and their ability of eliminating ADTs (and lists, 
in particular) very much depends on the  algorithm used.

The {\em Elimination Algorithm}, proposed in previous work~\cite{De&18a}, 
allows the removal of ADT variables in many
non-trivial examples.
However, that algorithm is not successful in our case here 
because it is not able to discover  auxiliary properties,
such as Lemma~$(2)$ in our case, which are often needed during the transformation.
A similar limitation also applies to some tools that extend SMT solvers
with induction~\cite{ReK15,Un&17,Ya&19}. Indeed,
 those tools sometimes fail to discover 
the suitable lemmas (such as Lemma~$(2)$) that are needed for 
the inductive proofs.

The new ADT removal algorithm $\mathcal R$,
which we present in  this
paper (see Section~\ref{sec:Strategy}), extends the Elimination Algorithm by providing a technique for the automatic invention of
predicates that correspond to the suitable lemmas needed for
eliminating ADTs from sets of CHCs.
Indeed, when applied to the set {\it RevCls} of clauses, Algorithm~$\mathcal R$ introduces three new predicates
defined, respectively, by the following clauses:

\vspace{-1mm}

{\small
\begin{verbatim}
D1. new1(N0,N1,N2) :- append(Xs,Ys,Zs), reverse(Zs,Rs), len(Xs,N0),
                      len(Ys,N1), len(Rs,N2).
D2. new2(N1,N2) :- reverse(Zs,Rs), len(Zs,N1), len(Rs,N2).
D3. diff(X,N2,N21) :- snoc(Rs,X,R1s), len(R1s,N21), len(Rs,N2).
\end{verbatim}
}

\vspace{-1mm}

\noindent
Predicate {\tts new1} is defined by taking the conjunction of the atoms
occurring in the body of clause~{\tts 1}.
Predicate {\tts new2} is defined from the body of a clause
derived by unfolding clause~{\tts D1} with respect to the atom
{\tts append} (using clause~{\tts 2}).
The definition of predicate {\tts diff} is based on a more complex
mechanism as it is derived by matching clause~{\tts D1} against
the following clause {\tts D1$^{\textstyle *}$} 
obtained by unfolding clause~{\tts D1} %
with respect to the atoms {\tts append} (using clause~{\tts 3}), 
{\tts reverse}, and~{\tts len}:

\vspace{1mm}
\noindent
\makebox[42mm][l]{{\tts D1$^{\textstyle *}$\!\!. new1(N01,N1,N21) :- }}{\tts N01=N0+1, append(Xs,Ys,Zs), reverse(Zs,Rs),}

\noindent
\hspace{42mm}{\tts len(Xs,N0), len(Ys,N1), snoc(Rs,X,R1s),}

\noindent
\hspace{42mm}{\tts len(R1s,N21).}

\vspace{1mm}
\noindent
The definition of the predicate {\tts diff}, given by clause~{\tts D3}, comes from the mismatch 
between the bodies of  
clauses~{\tts D1} and~{\tts D1$^{\textstyle *}$} and,
for this reason, that predicate is called a {\em difference predicate}.
Indeed, the body of clause~{\tts D3} is made out of: (i)~the atoms
{\tts snoc(Rs,X,R1s)} and {\tts len(R1s,N21)},
which occur in clause~{\tts D1$^{\textstyle *}$} and do not occur in clause~{\tts D1},
and (ii)~the atom {\tts len(Rs,N2)},
which occurs in clause~{\tts D1} and does not occur in clause~{\tts D1$^{\textstyle *}$}.

In Section~\ref{sec:Strategy}, we will provide the formal definition of
Algorithm $\mathcal R$. We will also revisit the {\it Reverse} example
and we will give a detailed account on how the definitions of the 
predicates {\tts new1},\! {\tts new2},
and {\tts diff} can be introduced in a fully automatic way.

The transformation of {\it RevCls} together with the additional clauses
{\tts D1, D2,} and {\tts D3}, is now done according to
a routine application of the fold/unfold rules.
Indeed, we get the following final set \!\textit{TransfRevCls} of clauses
without list arguments (the numbering of the clauses refers to
the detailed trasformation shown in Section~\ref{sec:Strategy}):

\vspace{-1mm}

{\small
	\begin{verbatim}
	10. false :- N2=\=N0+N1, new1(N0,N1,N2).
	15. new1(N0,N1,N2) :- N0=0, new2(N1,N2).
	17. new1(N0,N1,N2) :- N0=N+1, new1(N,N1,M), diff(X,M,N2).
	18. new2(M,N) :- M=0, N=0.
	19. new2(M1,N1) :- M1=M+1, new2(M,N), diff(X,N,N1).
	20. diff(X,N0,N1) :- N0=0, N1=1.
	21. diff(X,N0,N1) :- N0=N+1, N1=M+1, diff(X,N,M).
	\end{verbatim}
}

\vspace{-1mm}

The Eldarica CHC solver proves the satisfiability of \textit{TransfRevCls} by
computing the following  {\textit{LIA}} model,
which we write in a Prolog-like syntax as a set of constrained facts:

\vspace{-2mm}

{\small
	\begin{verbatim}
	new1(N0,N1,N2) :- N2=N0+N1, N0>=0, N1>=0, N2>=0.
	new2(M,N) :- M=N, M>=0, N>=0.
	diff(X,N2,N21) :- N21=N2+1, N2>=0.
	\end{verbatim}
}

\vspace*{-2mm}

Note that, if in clause~{\tts D3}  we replace the atom
{\tts diff(N2,X,N21)} by its model computed by Eldarica, namely the constraint `{\tts N21=N2+1, N2>=0}',
we get the following formula:

\smallskip

\noindent
{\small
	$\mathtt{\forall}\!$ {\tt Rs,X,R1s,N21\!,N2.\,snoc(Rs,X,R1s)\!,\,len(R1s\!,N21)\!,\,len(Rs\!,N2)} $\mathtt{\rightarrow}$ {\tt N21=N2+1,N2>=0}
}

\smallskip

\noindent
which is equivalent to  Lemma~$(2)$.
Thus, in this case, the introduction of the difference predicate
performed by Algorithm~$\mathcal R$, can be
viewed as a way of automatically introducing the
lemma needed for constructing the inductive proof of 
Property (1).

\section{Constrained Horn Clauses}
\label{sec:CHCs}
In this section we recall some basic notions about CHCs.
Let \textit{LIA} be the theory of linear integer arithmetic and
\textit{Bool} be the theory of boolean values.
A {\em constraint\/} is a quantifier-free formula of $\textit{LIA}\cup\textit{Bool\/}$.
Let~$\mathcal C$ denote the set of all constraints.
Let~$\mathcal L$ be a typed first order language with equality~\cite{End72}
which includes the language of $\textit{LIA}\cup\textit{Bool\/}$.
Let $\textit{Pred}$ be a set of predicate
symbols in $\mathcal L$ not occurring in the language of $\textit{LIA}\cup\textit{Bool\/}$.

The integer and boolean types are said to be %
{\it basic types}.
Here, for reasons of simplicity, we do not consider other basic types,
such as real numbers, arrays, and bit-vectors, which are
usually supported by SMT solvers~\cite{CVC4,DeB08,HoR18}.
The \mbox{non-basic} types are collectively called
{\it algebraic data types} (ADTs), which %
are specified  by suitable data-type declarations such as the
{\small{\tt declare-datatypes}} declarations adopted by
SMT solvers.

An {\it atom} is a formula of the form $p(t_{1},\ldots,t_{m})$,
where~$p$ is a typed predicate symbol in $\textit{Pred}$, and
$t_{1},\ldots,t_{m}$ are typed terms constructed out of individual
variables, individual constants,  and function symbols.
A~{\it constrained Horn clause}  (or a CHC, or simply, a {\it clause}) is
an implication of the form
$H\leftarrow c, B$ (for clauses we use the logic programming notation, where
comma denotes conjunction). The conclusion (or {\it head\/}) $H$
is either an atom or \textit{false},
the premise (or {\it body\/}) is the conjunction of
a constraint  $c\!\in\!\mathcal{C}$, and a (possibly empty) conjunction~$B$ of atoms.
If the head $H$ of a clause is an atom of the form $p(t_{1},\ldots,t_{n})$, the predicate $p$
is said to be a {\it head predicate}.
A clause whose head is an atom is called a {\it definite clause},
and a clause whose head is {\it false} is called a {\it  goal\/}.

We assume  that, for every atom $A$ occurring in a clause,
(i)~each term of basic type occurring in $A$ is a variable, and
(ii)~no variable of basic type occurs in $A$ more than once.
For instance, the atom {\small{\tt p(X,[Y\,|\,T])}} may occur in a clause,
while by Condition~(i), the atoms {\small{\tt p(3,[Y\,|\,T])}}
and {\small{\tt p(X,[Y+Z\,|\,T])}} may not.
Conditions~(i) and (ii) on atoms can always be enforced
at the expense of introducing new variables subject to constraints in the body of the clause.
These conditions ensure that, when applying the  unfolding rule (see Section~\ref{sec:TransfRules}),
the unification of terms of basic type
can be delegated to constraint solving.

We assume that all variables in a clause are universally quantified in front, and thus
we can freely rename them.
Clause $C$ is said to be a {\it variant}
of clause $D$ if $C$ can be obtained from $D$ by renaming variables
and rearranging the order of the atoms in its body.
Given a term~$t$, by ${\it vars}(t)$ we denote the set of all
variables occurring in $t$.
Similarly, for the set of all variables occurring in a formula or a set of formulas.
Given a formula $\varphi$ in ${\mathcal L}$, we denote by
$\forall (\varphi)$ its {\it universal closure}.

Let~$\mathbb D$ be the usual interpretation for the symbols in
$\textit{LIA}\cup\textit{Bool\/}$, and let a~\mbox{\it $\mathbb D$-interpretation} be an interpretation of~$\mathcal L$
that  agrees with
$\mathbb D$, for all symbols occurring in~$\textit{LIA}
\cup\textit{Bool\/}$.
A $\mathbb D$-model of a clause~$C$ is
a $\mathbb D$-interpretation that makes $C$ true.
A $\mathbb D$-model of a set $P$ of clauses is a
$\mathbb D$-model of every clause in~$P$. The reference to the
interpretation $\mathbb D$ will be omitted when it is irrelevant or
understood from the context.

A set $P$ of CHCs is said to be {\it $\mathbb D$-satisfiable}
(or {\it satisfiable}, for short) if
it has a \mbox{$\mathbb D$-model,} and it is said to be
{\em $\mathbb D$-unsatisfiable} (or {\it unsatisfiable},
for short), otherwise.
Given two {${\mathbb D}$-interpretations $\mathbb I$ and $\mathbb J,$
we say that $\mathbb I$ is {\em included} in $\mathbb J$
if for all ground atoms~$A$, $\mathbb I\models A$ implies $\mathbb J\models A$.}
Every set $P$ of definite  clauses is {satisfiable} and
has a {\it least}
(with respect {to inclusion}) ${\mathbb D}$-model,
denoted $M(P)$, which is equal to the set of all ground atoms that are true in
all ${\mathbb D}$-models of~$P$~\cite{JaM94}.
If $P$ is any set of constrained Horn clauses and
$Q$ is the set of the goals in~$P$, then we define
$\textit{Definite}(P)$ to be the set $P \setminus Q$.
It is the case that~$P$ is satisfiable
if and only if $M(\textit{Definite}(P))\models Q$.
We will often use a variable as an argument of a predicate
to actually denote a tuple of variables. For instance,
we will write $p(X,Y)$, instead of
$p(X_{1},\ldots,X_{m},Y_{1},\ldots,Y_{n})$,
whenever the %
values of \mbox{$m~(\geq\! 0)$} and \mbox{$n~(\geq\! 0)$} are not relevant.
Whenever the order of the variables is not relevant,
we will feel free to identify %
tuples of distinct variables
with finite sets.

We will also extend to finite tuples the operations and relations
which are usually defined on  sets.
Given two tuples $X$ and~$Y$ of distinct variables, %
(i)~their {\it union} $X\cup Y$ is obtained
by concatenating them and removing all duplicated occurrences of variables,
(ii)~their {\it intersection} $X\cap Y$ is obtained
by removing from~$X$ the variables which do not occur in~$Y$,
(iii)~their {\it difference} $X \!\setminus \!Y$ is obtained
by removing from $X$ the variables which occur in $Y$, and
(iv)~$X\!\subseteq\! Y$ holds if every variables of $X$ occurs in $Y$.
For all $m\!\geq\!0$, equality of $m$-tuples of terms is defined as follows:
$(u_{1},\ldots,\!u_{m}) = (v_{1},\!\ldots,\!v_{m})$
iff $\bigwedge_{i=1}^{m} (u_{i}\!=\!v_{i})$. The empty tuple~$()$ is identified with the empty set $\emptyset$.

By $A(X,Y)$, where $X$ and $Y$ are disjoint tuples of {distinct}
variables, we denote an atom $A$ such that $\mathit{vars}(A) = X\cup Y$.
Given the atom $A(X,Y)$, if $Y$ is the $k$-tuple $(Y_{1},\ldots,Y_{k})$
of distinct variables, and $Z$ is the $k$-tuple $(Z_{1},\ldots,Z_{k})$
of distinct variables not occurring in $X$, then
by $A(X,Z)$ we denote the atom obtained by replacing in $A(X,Y)$ the variable
$Y_{i}$ by $Z_{i}$, for $i\!=\!1,\ldots,k$.
The atom $A(X,Y)$
is said to be {\em functional from the input variables} $X$ {\em to the output variables} $Y$ {\em with respect to the set $P$ of definite clauses} if
\smallskip

$(\mathit{Funct})$ ~~$M(P) \models \forall X, Y, Z.\ A(X,Y) \wedge A(X,Z) ~\rightarrow~ Y\!=\!Z$

\smallskip
\noindent
The atom $A(X,Y)$
is said to be {\em total from the input variables} $X$ {\em to the
output variables}~$Y$ {\em with respect to the set $P$ of definite clauses} if

\smallskip

$(\mathit{Total})$  ~~$M(P) \models \forall X \exists Y.\ A(X,Y)$

\smallskip
\noindent
If $A(X,Y)$ is a total, functional atom from $X$ to $Y$,
we will also write $A(X;Y)$.
For instance, with respect to the definite
clauses {\tts 2}--{\tts 7} shown in
Section~\ref{sec:IntroExample}, we have that:
(i)~{\tts append(Xs,Ys,Zs)} is a total, functional atom from
the pair {\tts (Xs,Ys)} of input variables
to the output variable {\tts Zs}, and (ii)~{\tts reverse(Zs,Rs)} is a total,
functional atom from {the input variable}~{\tts Zs} to {the output variable}~{\tts Rs}.

When referring to the notions of functionality and totality,
we will feel free not to mention the set $P$ of definite clauses,
if it is understood from the context.

Note that, in our application to program verification,
the initial set of clauses is obtained
by translating a terminating functional program into CHCs,
and hence the functionality Property (\textit{Funct}) and
the totality Property  (\textit{Total})
hold by construction for any
atom $A(X,Y)$ that translates a function defined by that program.

\indent
We can extend the functionality and totality notions
from atoms to conjunctions of atoms as follows.
Let $F(X,Y)$ denote a conjunction $A_1, \ldots, A_n$
of~$n\,(\geq\!1)$ atoms, where $X$ and $Y$ are disjoint tuples
of distinct variables such that \mbox{$\mathit{vars}(\{A_1, \ldots,A_n\})
\!=\! X \cup Y\!$.}
Then, $F(X,Y)$ is said to be functional from~$X$ to $Y$ if Property~$(\mathit{Funct})$
holds for $F(X,Y)$ and $F(X,Z)$, instead of
$A(X,Y)$ and $A(X,Z)$, respectively. Similarly,
$F(X,Y)$ is said to be total from $X$ to $Y$ if Property~$(\mathit{Total})$ holds
for $F(X,Y)$, instead of  $A(X,Y)$.
If $F(X,Y)$ is a total, functional conjunction from $X$ to~$Y$,
we will also write $F(X;Y)$.

Now, let us consider a conjunction $F$ of $n\,(\geq\!1)$ total, functional atoms, which
(modulo reordering) is equal to `$A_1(X_1;Y_1), \ldots,A_n(X_n;Y_n)$',
where: (1)~the output (tuples of) variables $Y_{i}$'s are pairwise disjoint, and
(2)~for $i\!=\!1,\ldots,n$, $(\bigcup^i_{j=1} X_j)
\cap Y_i  = \emptyset$.
Then, $F$ is a total, functional conjunction from $X$ to $Y$,
where $Y\!=\!\bigcup^n_{i=1} Y_i$ and
$X\!=\!(\bigcup^n_{i=1} X_i)\!\setminus\!Y$, and hence it  can be denoted by $F(X;Y)$.
For instance,
the conjunction `{\tts append(Xs,\!Ys,Zs),\,reverse(Zs,Rs)}', whose predicates are
defined %
in Section~\ref{sec:IntroExample}, is
a total, functional conjunction from the pair {\tts (Xs,Ys)} of input variables to the pair {\tts (Zs,Rs)}
of output variables.

\section{Transformation Rules for Constrained Horn Clauses}
\label{sec:TransfRules}
In this section we present the rules that we
propose for transforming CHCs, and in particular,
for introducing difference predicates, and
we prove the soundness of those rules.

\subsection{The transformation rules}
\label{subsec:Rules}

First, we introduce the following notion of a {\it stratification} for a set of clauses. 
Let $\mathbb N$  denote the set of the natural numbers.
A \emph{level mapping} is a 
function~$\ell\!:\mathit{Pred}\!\rightarrow\!\mathbb{N}$. For every predicate $p$,
the natural number $\ell(p)$  is said to be the {\it level\/} of~$p$.
Level mappings are extended to atoms by stating that 
the level $\ell(A)$ of an atom $A$ is the 
level of its predicate symbol.
A clause \( H\leftarrow c, A_{1}, \ldots, A_{n}
\) is {\it stratified with respect to the level mapping}~$\ell$ if, 
for \( i\!=\!1,\ldots ,n \),
\(\ell (H)\geq \ell(A_i)\).
A set $P$ of CHCs is {\it stratified~with respect to $\ell$} if all clauses
of $P$ are stratified with respect to~$\ell$.
Clearly, for every set~$P$ of CHCs, there exists a level mapping~$\ell$ such that
$P$ is stratified with respect to $\ell$~\cite{Llo87}. 

A {\it transformation sequence from} $P_{0}$ {\it to} $P_{n}$
is a sequence 
$P_0 \Rightarrow P_1 \Rightarrow \ldots \Rightarrow P_n$ of sets of CHCs 
such that, for $i\!=\!0,\ldots,n\!-\!1,$ $P_{i+1}$ is derived from $P_i$, denoted
$P_{i} \Rightarrow P_{i+1}$, by
applying one of the following \mbox{rules~R1--R7}. We assume that the initial set
$P_0$ is stratified with respect to~a given level mapping~$\ell$.

\medskip

The Definition Rule allows us to introduce new predicate definitions.
\medskip
\hrule
\vspace*{1.5mm}
\noindent
{\bf(R1)~ Definition Rule.}  
Let $D$ be the clause $\textit{newp}(X_1,\ldots,X_k)\leftarrow 
c,A_1,\ldots\!,A_m$, where:
(1)~\textit{newp} is a predicate symbol in $\textit{Pred\/}$ not occurring in 
the sequence $P_0\Rightarrow P_1\Rightarrow\ldots\Rightarrow P_i$ constructed 
so far, \mbox{(2)~$c$ is a constraint,} 
(3)~the predicate symbols of $A_1,\ldots,A_m$
occur in $P_0$, and 
(4)~$(X_1,\ldots,X_k)\subseteq \mathit{vars}(\{c,A_1,\ldots,A_m\})$.

Then, by {\it definition} we get $P_{i+1}= P_i\cup \{D\}$. 
We define the level mapping $\ell$ of \textit{newp} to be equal 
to $\textit{max}\,\{\ell(A_i) \mid i=1,\ldots,m\}$.

\smallskip
\hrule

\medskip

For $j\!=\!0,\ldots, n$, by $\textit{Defs}_j$ 
we denote the set of clauses, 
called {\it definitions}, 
introduced by rule~R1 during the construction of the prefix 
$P_0\Rightarrow P_1\Rightarrow\ldots\Rightarrow P_j$ of the transformation sequence 
$P_0\Rightarrow P_1\Rightarrow\ldots\Rightarrow P_n$. 
Thus, $\textit{Defs}_0\!=\!\emptyset$, and for $j\!=\!0,\ldots, n,$ 
$\textit{Defs}_j\!\subseteq\!\textit{Defs}_{j+1}$.
Note that, by using rules R2--R7, one may replace 
a definition occurring in~$P_h$,
for some $0\!<\!h\!<\!n$, and hence it may happen 
that $\textit{Defs}_{k}\!\not\subseteq\!P_{k}$, 
for some $k$ such that $h\!<\!k\!\leq\!n$.

\medskip

The Unfolding Rule consists in performing a symbolic computation step.

\medskip
\hrule
\vspace*{1.5mm}
\noindent
{\bf (R2)~Unfolding Rule.} 
Let  $C$: $H\leftarrow c,G_L,A,G_R$ be a clause in $P_i$, where $A$ is an atom.
Without loss of generality, we assume that
$\mathit{vars}(C)\cap\mathit{vars}(P_0)=\emptyset$.
Let {\it Cls}: $\{K_{1}\leftarrow c_{1},
B_{1},~\ldots,~K_{m}\leftarrow c_{m}, B_{m}\}$, with $m\!\geq\!0$,
be the set of clauses in $P_0$,
such that: for $j=1,\ldots,m$,
(1)~there exists a most general unifier~$\vartheta_j$ of $A$ and
$K_j$, and {(2)~the conjunction of constraints $(c, c_{j})\vartheta_j$ is satisfiable.}
Let $\mathit{Unf}(C,A,P_0)$ be the set $\{(H\leftarrow  c, {c}_j,G_L, B_j, G_R) 
\vartheta_j \mid  j=1, \ldots, m\}$ of clauses.

Then, by {\it unfolding~$C$ with respect to $A$}, we derive the set 
$\mathit{Unf}(C,A,P_0)$ %
and we get $P_{i+1}= (P_i\setminus\{C\}) \cup \mathit{Unf}(C,A,P_0)$.
\smallskip
\hrule

\medskip

When we apply rule R2, we say that, for $j=1, \ldots, m,$ the atoms in the conjunction 
$B_j \vartheta_j$ are {\it derived} from $A$, and the atoms 
in the conjunction $(G_L, G_R) \vartheta_j$ are {\it inherited} from the corresponding atoms in the body of $C$.

\medskip

The Folding Rule is a special case of an inverse of the Unfolding Rule.
\nopagebreak

\medskip
\hrule
\vspace*{1.5mm}
\noindent
{\bf (R3)~Folding Rule.} 
Let $C$: $H\leftarrow c, G_L,Q,G_R$ be a clause in $P_i$, 
and let
$D$: $K \leftarrow d, B$ be a variant of a clause in $\textit{Defs}_i$.
Suppose that:
(1)~either $H$ is $\mathit{false}$ or \mbox{$\ell(H) \geq \ell(K)$,} and
(2)~there exists a substitution~$\vartheta$ such that~\mbox{$Q\!=\! B\vartheta $} and
$\mathbb D\models \forall(c \rightarrow d\vartheta)$.

Then,\,by \textit{folding \( C\)\,using definition\,\( D\)}, we derive clause 
\(E  \):~\( H\leftarrow c, G_{\!L}, K\vartheta, G_{\!R} \), and we get
\( P_{i+1}= (P_{i}\setminus\{C\})\cup \{E \} \).

\smallskip
\hrule

\medskip

The Clause Deletion Rule removes a clause with an unsatisfiable constraint in its body.

\medskip
\hrule
\vspace*{1.5mm}
\noindent
{\bf (R4)~Clause Deletion Rule.} 
Let  $C$: $H\leftarrow c,G$ be a clause in $P_i$ such that the constraint~$c$ is 
unsatisfiable. 

Then, by {\it clause deletion} we get
$P_{i+1} = P_i \setminus\{C\}$.
\smallskip
\hrule

\medskip

The Functionality Rule rewrites a functional conjunction of atoms by using Property ({\it Funct}).

\medskip
\hrule
\vspace*{1.5mm}
\noindent
{\bf (R5)~Functionality Rule.} 
Let $C$: $H\leftarrow c, G_L,F(X,Y),F(X,Z), G_R$ be a clause in~$P_i$,
where $F(X,Y)$ is a functional conjunction of atoms {from~$X$ to~$Y$} with respect to
$\mathit{Definite}(P_0) \cup \mathit{Defs}_i$.

Then, by \textit{functionality}, from~$C$ we derive~$D$: $H\!\leftarrow c, Y\!\!=\!Z, G_{\!L},F(X,\!Y),G_{\!R}$, and
we get \( P_{i+1}= (P_{i}\setminus\{C\})\cup \{D \} \).
\smallskip
\hrule

\medskip

The Totality Rule rewrites a functional conjunction of atoms by using Property~({\it Total\/}).

\medskip

\hrule
\vspace*{1.5mm}
\noindent
{\bf (R6)~Totality Rule.} 
Let $C$: $H\leftarrow c,  G_L,F(X,Y),G_R$ be a clause in $P_i$
such that $Y \cap \mathit{vars}(H\leftarrow c,  G_L,G_R) = \emptyset$ and
$F(X,Y)$ is a total conjunction of atoms {from~$X$ to~$Y$} with respect to
$\mathit{Definite}(P_0)\cup \mathit{Defs}_i$.

Then, by \textit{totality}, from~$C$ we derive clause~$D$\,: $H\leftarrow c, G_L,G_R$,
and we get \( P_{i+1}= (P_{i}\setminus\{C\})\cup \{D \} \).
\smallskip
\hrule

\medskip

As mentioned above,
the functionality and  totality properties hold by construction,
and we do not need to prove them when applying rules~R5 and~R6.

\medskip

The Differential Replacement Rule replaces a conjunction of atoms 
by a new conjunction together with an atom defining
a relation among the variables of those conjunctions.

\medskip

\hrule
\smallskip
\vspace*{1.5mm}
\noindent
{\bf (R7)~Differential Replacement Rule.} 
Let $C$: $H\leftarrow c,  G_L,F(X;Y),G_R$ be a clause in $P_i$, and
let $D$: $\mathit{diff}(Z) \leftarrow d, F(X;Y), R(V;W)$
be {a variant of} a definition clause in $\mathit{Defs}_i$, such that: 
\noindent
(1)~$F(X;Y)$ and $R(V;W)$ are total, functional conjunctions 
with respect to $\mathit{Definite}(P_0)\cup \mathit{Defs}_i$,
(2)~$W\cap \vars(C)\! =\!\emptyset$, 
(3)~$\mathbb{D}\models\forall (c\!\rightarrow\! d)$, and
(4)~\mbox{$\ell(H)\!>\!\ell(\!\mathit{diff(Z)})$}. 

Then, by {\it differential replacement},
we derive clause~$E$: $H\!\leftarrow \!c, G_{\!L},R(V;\!W),$ $\mathit{diff}(Z), G_{\!R}$,
and we get $P_{i+1}= (P_{i}\setminus\{C\}) \cup \{E \}$. 
\smallskip
\hrule

\medskip

Note that in rule~R7 no assumption is made on the set $Z$ of
 variables, apart from the one deriving from the fact that $D$ is a 
 definition, that is, 
 $Z\! \subseteq\! {\mathit{vars}}(d) \cup X\cup Y \cup V \cup W.$ 

The transformation algorithm~\Diff~for the removal of ADTs,
which we will present in Section~\ref{sec:Strategy}, 
applies a specific instance of rule~R7
(see, in particular, the Diff-Introduce step). 
The general form of rule~R7 that %
we have {now considered,} 
makes it 
easier to prove the Soundness and Completeness Theorems 
(see Theorems~\ref{thm:unsat-preserv} and~\ref{thm:sat-preserv}) we will present below.

\subsection{Soundness of the transformation rules}
\label{subsec:soundness}

Now we will extend to rules \mbox{R1--R7} 
some correctness results that have been proved for the transformation of
(constraint) logic programs~\cite{EtG96,Fi&04a,Sek09,TaS86}. 

\begin{theorem}[Soundness of the Transformation Rules]
\label{thm:unsat-preserv}
Let $P_0 \Rightarrow P_1 \Rightarrow \ldots \Rightarrow P_n$ be 
a transformation sequence
using rules {\rm{R1--R7}}.
Suppose that the following condition holds\,$:$
	
\vspace{.5mm}
\noindent\hangindent=8mm
\makebox[8mm][l]{\rm \,(U)}for $i \!=\! 1,\ldots,n\!-\!1$, if $P_i \Rightarrow P_{i+1}$ by 
folding a clause in $P_i$ using a definition $D\!: H \leftarrow c,B$  
in $\mathit{Defs}_i$, 
then, for some $j\! \in\!\{1,\ldots,i\!-\!1,i\!+\!1,\ldots, n\!-\!1\}$, 
$ P_{j}\Rightarrow P_{j+1}$ by unfolding
$D$ with respect to an atom $A$ such that $\ell(H)=\ell(A)$. 
	
\vspace{.5mm}
\noindent
If	$P_n$ is satisfiable, then $P_0$ is satisfiable.
\end{theorem}

Thus, to prove the satisfiability of a 
set~$P_0$ of clauses,
it suffices to: (i)~construct a transformation sequence  
$P_0 \Rightarrow P_1 \Rightarrow \ldots \Rightarrow P_n$, 
and then (ii)~prove that $P_n$ is satisfiable.

The need for Condition (U) in Theorem~\ref{thm:unsat-preserv} can be shown
by the following example.

\vspace{-1mm}
\begin{example}\label{ex:need-of-Cond-U}
Let us consider the following initial set of clauses:

\vspace{1mm}
\begin{minipage}[t]{8mm}
$P_{0}$:
\end{minipage}
\begin{minipage}[t]{80mm}
{\small{
\begin{verbatim}
1. false :- p.
2. p.
\end{verbatim}
}}
\end{minipage}

\vspace{1mm}
\noindent
By rule~R1 we introduce the definition:

\vspace{.5mm}
\hspace{9mm}{\tts 3. newp :- p.}

\vspace{.5mm}
\noindent
and we get the set $P_{1}\!=\!\{${\tts 1,2,3}$\}$ of clauses.
Then, by folding clause~{\tts 1} using definition~{\tts 3}, we get:

\vspace{1mm}
\begin{minipage}[t]{8mm}
$P_{2}$:
\end{minipage}
\begin{minipage}[t]{80mm}
{\small{
\begin{verbatim}
1f. false :- newp.
2.  p.
3.  newp :- p.
\end{verbatim}
}}
\end{minipage}

\vspace{1mm}
\noindent
Again, by folding  definition~{\tts 3} using the same 
definition~{\tts 3}, we get:

\vspace{1mm}
\begin{minipage}[t]{8mm}
$P_{3}$:
\end{minipage}
\begin{minipage}[t]{80mm}
{\small{
\begin{verbatim}
1f. false :- newp.
2.  p.
3f. newp :- newp.
\end{verbatim}
}}
\end{minipage}
\vspace{1mm}

\noindent
Now we have that $P_{3}$ is satisfiable (being \{{\tts p}\} its least $
{\mathbb D}$-model), while $P_{0}$ is 
unsatisfiable. This fact is consistent with Theorem~\ref{thm:unsat-preserv}. Indeed, 
the transformation sequence 
$P_{0}\Rightarrow P_{1}\Rightarrow P_{2}\Rightarrow P_{3}$ does not
comply with Condition~(U) because during that sequence,  definition~{\tts{3}}
has not been unfolded.  \hfill $\Box$
\end{example}

The following example shows that for the application of rule~R7, 
Condition~(4) cannot be dropped
because, otherwise, Theorem~\ref{thm:unsat-preserv} does not hold.

\vspace{-1mm}
\begin{example}
Let us consider the initial set of clauses:

\vspace{1mm}
\begin{minipage}[t]{8mm}
$P_{0}$:
\end{minipage}
\begin{minipage}[t]{80mm}
{\small{
\begin{verbatim}
1. false  :- r(X,Y).
2. r(X,Y) :- f(X,Y).
3. f(X,Y) :- Y=0.
\end{verbatim}
}}
\end{minipage}
\vspace{1mm}

\noindent where {\tts f} and {\tts r} are predicates whose arguments are in the 
set~$\mathbb Z$ of the integers.
Let us assume that the level mapping $\ell$
is defined as follows: 
$\ell(${\tts f}$)\!=\!1$ and $\ell(${\tts r}$)\!=\!2$.
Now, we apply rule R1 and we introduce a new predicate~{\tts diff}, and we get:

\vspace{1mm}
\begin{minipage}[t]{8mm}
$P_{1}$:
\end{minipage}
\begin{minipage}[t]{80mm}
{\small{
\begin{verbatim}
1. false  :- r(X,Y).
2. r(X,Y) :- f(X,Y).
3. f(X,Y) :- Y=0.
4. diff(X,W,Y) :- f(X,Y), r(X,W).
\end{verbatim}
}}
\end{minipage}
\vspace{1mm}

\noindent where, complying with rule R1, we set $\ell(${\tts{diff}}$)\!=\!2$.
By applying rule~R7, even if Condition~(4) is not satisfied, we get:

\vspace{1mm}
\begin{minipage}[t]{8mm}
$P_{2}$:
\end{minipage}
\begin{minipage}[t]{80mm}
{\small{
\begin{verbatim}
1.  false  :- r(X,Y).
2r. r(X,Y) :- r(X,W), diff(X,W,Y).
3.  f(X,Y) :- Y=0.
4.  diff(X,W,Y) :- f(X,Y), r(X,W).
\end{verbatim}
}}
\end{minipage}
\vspace{1mm}

\noindent
Now, contrary to the conclusion of Theorem~\ref{thm:unsat-preserv},
we have that $P_{0}$ is unsatisfiable and $P_{2}$ is satisfiable,
being \{{\tts f(n,0)} 
$\mid$ {\tts n}\,$\in\!{\mathbb Z}\}$ its least
${\mathbb Z}$-model.
Note that the other Conditions~(1), (2), and~(3) for applying rule~R7 do hold. In particular, 
the atoms {\tts f(X,Y)} and {\tts r(X,Y)} are total, functional atoms
from~{\tts X} to~{\tts Y} with respect to $P_{0}$. 
\hfill $\Box$
\end{example}

\medskip

The rest of this section is devoted to the proof of Theorem~\ref{thm:unsat-preserv}.

\newcommand{\Iff}{\leftrightarrow}
\newcommand{\ONLYIF}{\Rightarrow}
\newcommand{\IFF}{\Leftrightarrow}

First, we recall and recast in our framework
some definitions and facts taken from the literature~\cite{EtG96,TaS84,TaS86}.
Besides the rules presented in Section~\ref{subsec:Rules} above,
let us also consider the following rule R8 which, given a transformation sequence 
$P_0 \Rightarrow P_1 \Rightarrow \ldots \Rightarrow P_i$, for some $i\!\geq\!0$, 
allows us to 
extend it by constructing a new set $P_{i+1}$ of CHCs such as $P_i\Rightarrow P_{i+1}$.

\medskip

\hrule
\vspace*{1.5mm}
\noindent 
\textbf{(R8)~Goal Replacement Rule.} Let  
$C$:~${H}\leftarrow c, c_{1}, {G}_{L},
{G}_{1}, {G}_{R}$ be a clause in~${P}_{i}$.
If in 
clause~$C$ we  \emph{replace} $c_1,G_1$ by $c_2,G_2$, 
we derive  clause~$D$:
${H}\leftarrow c, c_{2}, {G}_{L},
{G}_{2}, {G}_{R}$, and we get 
$P_{i+1} = ({P}_{i}\setminus\{C\})\cup \{D\}$. 
\smallskip
\hrule

\medskip

Now let us introduce two particular kinds of Goal Replacement Rule: 
(i)~{\it{body weakening}}, and (ii)~{\it{body strengthening}}.
First, for the Goal Replacement Rule, we need to consider the following two tuples of variables:

\vspace*{1mm}
$T_1 = \mathit{vars}(\{c_{1},{G}_{1}\}) 
\setminus \mathit{vars}(\{{H}, c, {G}_{L},{G}_{R}\})$,~~ and

$T_2 = \mathit{vars}(\{c_{2},{G}_{2}\}) 
\setminus \mathit{vars}(\{{H}, c, {G}_{L},{G}_{R}\})$. 

\vspace*{1mm}
\noindent
The Goal Replacement Rule is said to be a {\em body weakening} if the
following two conditions hold:

\vspace*{1mm}

\makebox[12mm][r]{(W.1)~~} $M(\textit{Definite}(P_0)\cup \mathit{Defs}_i) \models \forall\,( c_{1}\! \wedge \! {G}_{1}
\rightarrow \exists T_2.\, c_{2}\! \wedge \! {G}_{2})$

\makebox[12mm][r]{(W.2)~~} $\ell(H)\!>\!\ell(\mathit{A})$, for every atom $A$ occurring 
in $G_2$ and not in $G_1$.

\smallskip
\noindent 
\noindent The Goal Replacement Rule is said to be a \emph{body strengthening} if the
following condition holds:

\vspace*{1mm}
\makebox[12mm][c]{(S)~~} $M(\textit{Definite}(P_0)\cup \mathit{Defs}_i) \models \forall \, (c_{2}\! \wedge \! {G}_{2}
\rightarrow \exists T_1.\, c_{1}\! \wedge \! {G}_{1})$.

\vspace*{1mm}

\smallskip

Usually, in the literature
the Goal Replacement Rule is presented by considering the conjunction 
of Conditions~(W.1) and~(S), thereby considering the quantified equivalence
$\forall \, ( (\exists T_1.\, c_{1}\! \wedge \! {G}_{1})
\Iff (\exists T_2.\, c_{2}\! \wedge \! {G}_{2}))$.
We have split that equivalence into the two associated implications.
This has been done because, 
when proving the Soundness result (see Theorem~\ref{thm:unsat-preserv} in this section) and the Completeness result 
(see Theorem~\ref{thm:sat-preserv}  in Section~\ref{subsec:completeness}),
it is convenient to present %
the preservation
of the least $\mathbb D$-models %
into two parts as specified by the following two theorems.

\begin{theorem}
\label{thm:cons} Let $D_0, \ldots, D_n$ be sets of definite CHCs
and let \mbox{$D_0 \repl \ldots \repl D_n$}
be a transformation sequence constructed using rules
{\rm R1} $($Definition$)$, {\rm R2} $($Unfolding$)$, {\rm R3} $($Folding$)$, and 
{\rm R8} $($Goal Replacement$)$. Suppose that
Condition~{\rm (U)} of Theorem~$\ref{thm:unsat-preserv}$ holds and all 
goal replacements are body weakenings.
Then $M({D}_0\cup \mathit{Defs}_n)\subseteq M({D}_n)$.
\end{theorem}

\begin{theorem}
	\label{thm:pc} Let $D_0, \ldots, D_n$ be sets of definite CHCs
	and let \mbox{$D_0 \repl \ldots \repl D_n$}
	be a transformation sequence constructed using rules
	{\rm R1} $($Definition$)$, {\rm R2} $($Unfolding$)$, {\rm R3} $($Folding$)$, and 
	{\rm R8} $($Goal Replacement$)$. Suppose that, for all applications of~{\rm R3}, Condition {\rm (E)} 
	holds $($see Definition~$\ref{def:condR}$ in Section~$\ref{sec:Completeness}$$)$
	and all goal replacements are 
	body strengthenings.
	Then \(M( D_0\cup \mathit{Defs}_n)\supseteq M( D_n) \).
\end{theorem}

For the proof of Theorems~\ref{thm:cons} and~\ref{thm:pc} we refer 
to the results presented in the literature~\cite{EtG96,TaS84,TaS86}.
The correctness of the transformation rules with respect to 
the least Herbrand model semantics has been first proved in the landmark
paper by Tamaki and Sato~\cite{TaS84}.
In a subsequent technical report~\cite{TaS86}, the same authors
extended that result 
by introducing the notion of the {\it level\/} of an atom, 
which we also use in this paper (see the notion defined 
at the beginning of Section~\ref{subsec:Rules}). 

The use of atom levels allows less restrictive 
applicability conditions on the Folding and Goal Replacement Rules.
Later, Etalle and Gabbrielli~\cite{EtG96} extended Tamaki and Sato's results 
to the \mbox{$\mathbb D$-model} semantics of
{\em constraint logic programs} (in the terminology used in
this paper, a constraint logic program is a 
set of definite constrained Horn clauses).

There are three main differences between our presentation of the correctness 
results for the transformation rules with respect to 
the presentation considered in %
the literature~\cite{EtG96,TaS84,TaS86}.

First, as already mentioned, we kept the two Conditions~(E) and~(S), 
which guarantee
the inclusion \(M(D_0\!\cup\! \mathit{Defs}_n)\!\supseteq\! M( D_n) \) (called 
{\em Partial Correctness} by Tamaki and Sato~\cite{TaS86}),
separated from the three Conditions~(U), (W.1), and~(W.2), which guarantee
the reverse inclusion \(M(D_0\cup\mathit{Defs}_n)\!\subseteq\! M( D_n) \).
All five conditions together garantee the equality
\(M(D_0\!\cup\! \mathit{Defs}_n)\!=\!M( D_n) \)
(called {\em Total Correctness} by Tamaki and Sato~\cite{TaS86}).

Second, Tamaki and Sato's conditions for the correctness of the Goal Replacement
Rule are actually more general than ours, as they
use a well-founded relation which is based on atom levels and
also on a suitable measure (called \mbox{\em weight-tuple measure}~\cite{TaS86})
of the successful derivations of an 
atom in $M( D_0\cup \mathit{Defs}_n)$. 
Our simpler conditions straightforwardly imply Tamaki and Sato's ones, 
and are sufficient for our purposes in the present paper.

Third, Tamaki and Sato papers~\cite{TaS84,TaS86} 
do not consider constraints, whereas 
Etalle and Gabbrielli results for constraint logic programs do not consider
Goal Replacement~\cite{EtG96}. However, Tamaki and Sato's proofs can easily be 
extended to
constraint logic programs by simply dealing with atomic constraints as atoms
with level 0 and assigning positive levels to all other atoms.

\medskip

From Theorem~\ref{thm:cons}, we get the following 
Theorem~\ref{thm:unsat},
which relates the satisfiability of sets of clauses obtained 
by applying the transformation rules to the satisfiability of the original sets of clauses.

\begin{theorem}
\label{thm:unsat} 
Let $P_0 \repl \ldots \repl P_n$
be a transformation sequence constructed using rules
{\rm R1} $($Definition$)$, {\rm R2} $($Unfolding$)$, {\rm R3} $($Folding$)$, and 
{\rm R8} $($Goal Replacement$)$. Suppose that
Condition~{\rm (U)} of Theorem~$\ref{thm:unsat-preserv}$ holds and all 
goal replacements are body weakenings.
If $P_n$ is satisfiable, then $P_0$ is satisfiable.
\end{theorem}

\begin{proof}
First, we observe that $P_0$ is satisfiable iff 
$P_0 \cup \textit{Defs}_n$ is satisfiable. Indeed, we have that: 
(i)~if~$\mathcal{M}$ is a $\mathbb D$-model of $P_0$, then the $\mathbb D$-interpretation
$\mathcal{M} \cup \{\textit{newp}(a_1,\ldots,a_k) \mid \textit{newp}$ 
is a head predicate
in $\textit{Defs}_n \mbox{ and } a_1,\ldots, a_k$ are ground terms$\}$
is a $\mathbb D$-model of $P_0 \cup \textit{Defs}_n$,
and (ii)~if $\mathcal{M}$ is a $\mathbb D$-model of $P_0 \cup \textit{Defs}_n$, then
all clauses of $P_0$ are true in $\mathcal{M}$, and hence
$\mathcal{M}$  is a $\mathbb D$-model of $P_0$.

Then, let us consider a new transformation 
sequence  $P'_0\Rightarrow \ldots \Rightarrow P'_n$
obtained from the  sequence $P_0\Rightarrow \ldots \Rightarrow P_n$
by replacing each occurrence of \textit{false} in the head of
a clause by a fresh, new predicate symbol, say $f$.
$P'_0,\ldots,P'_n$ are sets of definite clauses, and thus, for $i=0,\ldots,n,$ $\textit{Definite}(P'_i)= P'_i$.
The sequence \mbox{$P'_0\Rightarrow \ldots \Rightarrow P'_n$} satisfies the 
hypotheses of Theorem~\ref{thm:cons}, and hence 
$M(P'_0\cup \mathit{Defs}_n)\!\subseteq M(P'_n)$. We have that:    

\noindent
\makebox[12mm][l]{}$P_n$ is satisfiable

\noindent
\makebox[12mm][l]{implies}{$P'_n\cup \{\neg f\}$ is satisfiable}

\noindent
\makebox[12mm][l]{implies}{$f\not\in M(P'_n)$}

\noindent
\makebox[12mm][l]{implies,}\,by Theorem~\ref{thm:cons}, 
$f\not\in M(P'_0 \cup \textit{Defs}_n)$ 

\noindent
\makebox[12mm][l]{implies}{$P'_0 \cup \textit{Defs}_n\cup \{\neg f\}$ is satisfiable}

\noindent
\makebox[12mm][l]{implies}{$P_0 \cup \textit{Defs}_n$ is satisfiable}

\noindent
\makebox[12mm][l]{implies}{$P_0$ is satisfiable.} \hfill$\Box$
\end{proof}

Now, in order to prove Theorem~\ref{thm:unsat-preserv} of 
Section~\ref{sec:TransfRules}, which states the soundness of rules R1--R7, 
we show that rules R4--R7 are all body weakenings.

\medskip

An application of rule~R4~(Clause Deletion), by which we delete
clause~$C$: $H\leftarrow c,G$, 
whenever the constraint~$c$ is unsatisfiable, is equivalent to
the replacement of the body of clause $C$ by {\it false}. 
Since $c$ is unsatisfiable, we have that:

\smallskip
$M(\textit{Definite}(P_0)\cup \mathit{Defs}_i) \models \forall\,(c \wedge G \rightarrow {\it false})$

\smallskip
\noindent
and Condition (W.1) of rule R8 holds. Also Condition~(W.2), that is:

\vspace*{1mm}
$\ell(H)\!>\!\ell(\mathit{A})$, for every atom $A$ occurring in ${\it false}$ 

\vspace*{1mm}
\noindent
trivially holds, because there are no atoms in ${\it false}$.
Thus, the replacement of the body of clause $H\leftarrow c,G$ by {\it false}
is a body weakening.

\medskip 

Let us now consider rule~R5 (Functionality). 
Let $F(X,Y)$ be a conjunction of atoms 
that defines a functional relation from~$X$ to~$Y$, that is, Property (\textit{Funct})
of Section~\ref{sec:CHCs} holds for $F(X,Y)$.
When rule~R5 is applied whereby a conjunction $F(X,Y),F(X,Z)$
is replaced by the new conjunction $Y\!=\!Z,\ F(X,Y)$, we have that:

\vspace*{1mm}
$M(\mathit{Definite}(P_0)\cup \mathit{Defs}_i) \models \forall(F(X,Y) \wedge F(X,Z) \rightarrow Y\!=\!Z)$

\vspace*{1mm}
\noindent
and hence Condition~(W.1) of rule R8 holds. When this replacement is performed, also Condition~(W.2) trivially holds, and thus
rule~R5 is a body weakening. 

\medskip

An application of rule~R6 (Totality) replaces a conjunction $F(X,Y)$
by {\it true} (that is, the empty conjunction), which is implied by any formula. Hence
{Conditions (W.1) and (W.2) trivially hold,} %
and rule~R6 is a body weakening. 
\medskip

For rule~R7 (Differential Replacement) we prove the following lemma.
Recall that by $F(X;Y)$ we denote a conjunction of atoms 
that defines a total, functional relation from~$X$ to~$Y$.

\begin{lemma}\label{lemma:R7unsat} Let us consider a transformation sequence 
$P_{0}\Rightarrow \ldots \Rightarrow P_{i}$ and a 
clause $C${\rm :}~$H\leftarrow c,  G_L,F(X;Y),G_R$ in $P_{i}$. 
Let us assume that by applying rule~{\rm R7} on clause~$C$ using 
the definition clause

\smallskip
$D{\rm :}~\mathit{diff}(Z) \leftarrow d, F(X;Y), R(V;W),$

\smallskip
\noindent
where\,{\rm :} $(D1)$~$W \cap \vars(C) = \emptyset,$  and
$(D2)$~$\mathbb{D}\models\forall (c\rightarrow d),$ 
we derive clause

\smallskip
$E${\rm :} $H\leftarrow c,  G_L,R(V;W), \mathit{diff}(Z), G_R$

\smallskip
\noindent
and we get the new set 
$P_{i+1} = (P_{i}\setminus\{C\}) \cup \{E \}$ of clauses. Then,

\smallskip
$M(\textit{Definite}(P_0) \cup \mathit{Defs}_i)\models \forall( c \wedge F(X;Y) \rightarrow \exists W.\,(R(V;W) \wedge \mathit{diff}(Z)))$.
\end{lemma}

\begin{proof} Let $\mathcal M$ denote $M(\textit{Definite}(P_0) 
\cup \mathit{Defs}_i)$.
Since, %
$R(V;W)$ is a total, functional conjunction from $V$ to $W$
with respect to~$\textit{Definite}(P_0) \cup  \mathit{Defs}_i$,
we have:

\smallskip

$\mathcal M\models \forall\,( c \wedge F(X;Y) \rightarrow \exists W.\, R(V;W))$ \hfill $(\alpha)$~~~~

\smallskip
\noindent
Since, by Condition~$(D1)$, none of the variables in $W$ occurs in~$C$, from definition~$D$, we get:

\smallskip

$\mathcal M \models \forall\,(d \wedge F(X;Y) \wedge R(V;W) \rightarrow \mathit{diff}(Z))$\hfill $(\beta)$~~~~

   \smallskip

\noindent
From~$(\alpha)$, $(\beta)$, and Condition~$(D2)$, we get the thesis.\hfill $\Box$
\end{proof}

\medskip

From Lemma~\ref{lemma:R7unsat}
it follows that rule~R7, which replaces in the body of 
clause~$C${\rm :} $H\leftarrow c,  G_L,F(X;Y),G_R$
the conjunction~$F(X;Y)$ by the new conjunction $R(V;W)$, $\mathit{diff}(Z)$, 
is a body weakening, assuming that $\ell(H)\!>\!\ell(\mathit{diff}(Z))$.
Recall that, since clause $D$ is a definition clause, we  
have that $\ell(\mathit{diff}(Z))\!\geq\!\ell(\mathit{R})$, and thus we have
that $\ell(H)\!>\!\ell(\mathit{R})$.

The following lemma summarizes the facts we have shown above about rules~R4--R7.

\begin{lemma}\label{lemma:weakening} 
The applications of rules {\rm R4--R7} are all body weakenings.
\end{lemma}

Finally, having proved Lemma~\ref{lemma:weakening}, we can
 present the proof of Theorem~\ref{thm:unsat-preserv}.

\medskip

\noindent 
{\it Proof of Theorem~$\ref{thm:unsat-preserv}$}.
Let $P_0 \Rightarrow P_1 \Rightarrow \ldots \Rightarrow P_n$ be 
a transformation sequence constructed using rules {\rm{R1--R7}}. Then,
by Lemma~\ref{lemma:weakening},
that sequence
 can also be constructed by applications of rules {\rm{R1--R3}} together with
applications of rule~R8 which are all body weakenings. 
Since by hypothesis of Theorem~\ref{thm:unsat-preserv}, Condition~(U) does hold,
by applying Theorem~\ref{thm:unsat} we get the thesis.~\hfill~$\Box$
%
 

\section{An Algorithm for ADT Removal}
\label{sec:Strategy}

In this section we present Algorithm~\Diff~for 
eliminating ADT terms from CHCs 
by using the transformation rules R1--R7 presented in Section~\ref{sec:TransfRules}
and automatically introducing suitable 
difference predicates.
Then we show that, if Algorithm~\Diff~terminates, it transforms a set {\it Cls} 
of clauses into a new set $\mathit{TransfCls}$
where every argument of every predicate has a basic type.
Theorem~\ref{thm:unsat-preserv} (see Section~\ref{subsec:soundness}) guarantees that if $\mathit{TransfCls}$ is satisfiable,
then also {\it Cls} is satisfiable.  

\subsection{The ADT removal Algorithm~\Diff}
\label{subsec:algoR}

Algorithm~\Diff~(see Figure~\ref{fig:AlgoR}) 
removes ADT terms starting from the set~$\mathit{Gs}$ of 
goals in~$\mathit{Cls}$. Initially, those goals are all
collected in the set~$\mathit{InCls}$. 
The set~$\mathit{Defs}$ collects 
the  definitions of the new predicates
introduced by applications of rule~R1 during the execution of
Algorithm~\Diff. Initially, we have that $\mathit{Defs}\!=\!\emptyset$.

\begin{figure}[!ht]
\vspace{-5mm}
\noindent \hrulefill\nopagebreak

\noindent {\bf Algorithm}~\Diff\\
{\em Input}: A set $\mathit{Cls}$ of clauses and a level mapping $\ell$ of the
predicates occurring in $\mathit{Cls}$.
{\em Output}: A set $\mathit{TransfCls}$ of clauses that have basic types.

\vspace*{-2mm}
\noindent \rule{2.0cm}{0.2mm}

\noindent 
Let $\mathit{Cls} = \mathit{Ds} \cup \mathit{Gs}$, where $\mathit{Ds}$ is a set of definite clauses and $\mathit{Gs}$ 
is a set of goals;

\noindent $\mathit{InCls}:=\mathit{Gs}$;
\noindent $\mathit{Defs}:=\emptyset$;
\noindent $\mathit{TransfCls}:=\emptyset;$

\noindent
{\bf while} $\mathit{InCls}\!\neq\!\emptyset$ {\bf do}

\hspace*{3mm}\vline\begin{minipage}{11.8cm} %
\makebox[4mm][l]{~$\scriptscriptstyle\blacksquare$}$\mathit{Diff\mbox{-}Define\mbox{-}Fold}(\mathit{InCls},\mathit{Defs}, 
\mathit{NewDefs},\mathit{FldCls});$

\makebox[4mm][l]{~$\scriptscriptstyle\blacksquare$}$\mathit{Unfold}(\mathit{NewDefs},\mathit{Ds},\mathit{UnfCls});$

\makebox[4mm][l]{~$\scriptscriptstyle\blacksquare$}$\mathit{Replace}(\mathit{UnfCls}, \mathit{Ds}, \mathit{RCls});$

\makebox[4.5mm][l]{}$\mathit{InCls}:=\mathit{RCls};$~~~
$\mathit{Defs}:=\mathit{Defs}\,\cup\mathit{NewDefs};$~~~
$\mathit{TransfCls}:=\mathit{TransfCls}\,\cup\mathit{FldCls};$
\end{minipage} %

\vspace*{1mm} 
\noindent \hrulefill
\vspace*{-2mm} 
\caption{The ADT removal algorithm~\Diff \label{fig:AlgoR}.}
\vspace*{-6mm}
\end{figure}

\noindent
Algorithm~\Diff~iterates a sequence made out of the following three procedures.~

\smallskip
\noindent
(1)~{\it Procedure $\mathit{Diff\mbox{-}Define\mbox{-}Fold}$} introduces, by rule~R1,  
a set $\mathit{NewDefs}$ of suitable new predicate definitions. 
By applications of the Folding Rule~R3 and, possibly, of the Differential Replacement Rule~R7,
using clauses in $\mathit{Defs} \cup \mathit{NewDefs}$, the procedure removes
the ADT terms from the input set $\mathit{InCls}$ of clauses.

The bodies (but not the heads) of the clauses in $\mathit{NewDefs}$ contain ADT terms, and thus they
need to be transformed to remove those terms. 

\smallskip
\noindent
(2)~{\it Procedure $\mathit{Unfold}$} performs some steps of symbolic evaluation
of the newly introduced definitions by applying the Unfolding Rule~R2 to
the clauses occurring in $\mathit{NewDefs}$.

\smallskip
\noindent
(3)~{\it Procedure $\mathit{Replace}$} removes clauses that have an unsatisfiable body by applying rule~R4,
and also exploits the functionality and totality 
properties of the predicates by applying rules~R5 and~R6, respectively.

\smallskip
The clauses with ADTs obtained after the $\mathit{Replace}$ procedure
and the new predicate definitions introduced at each iteration,
are added to $\mathit{InCls}$ and $\mathit{Defs}$, respectively.
Algorithm~\Diff~terminates when the set~$\mathit{InCls}$ of clauses
becomes empty because no new definitions need to be introduced 
to perform folding steps.

\smallskip

Note that Algorithm~\Diff~takes as input also a level 
mapping~$\ell$ for the 
predicates occurring in {\it Cls}.
In our implementation, 
however, no function $\ell$ is actually provided and, instead,
a suitable level mapping is constructed during the execution of the 
algorithm itself. We do this construction %
by following a general constraint-based approach for guaranteeing 
the correctness of logic program transformations~\cite{Pe&12a}. 
In particular, given an initially empty
set~$L$ of constraints, each time~\Diff~applies a transformation rule
whose soundness depends on the satisfaction of a constraint
on the predicate levels (see, in particular the conditions
in rules~R1, R3, R7, and Condition (U) in Theorem~\ref{thm:unsat-preserv} for R2),
that constraint is added to the set~$L$. 
For the soundness of
Algorithm~\Diff, it is required that at the end of its execution,
the set $L$ be satisfiable. A solution of~$L$ provides 
the level mapping~$\ell$ to be constructed.
In order not to burden  %
the presentation with too many 
technical details, we will not present here the actual constraint 
handling mechanism used for the 
construction of the function~$\ell$.

\begin{example}[Reverse]\label{ex:rev1}
Throughout this section we will use the {\it Reverse} 
example of Section~\ref{sec:IntroExample} as a running example
for illustrating an application of the ADT removal algorithm~\Diff.
In that example, the set {\it Cls} of clauses given as input to~\Diff~consists
of clauses {\tts 1}--{\tts 9}, with ${\it Gs} = \{${\tts 1}$\}$
and ${\it Ds} = \{${\tts 2}$,\ldots,${\tts 9}$\}$.
Thus, ${\it InCls}$ is initialized to \{{\tts 1}\}.
We assume that the following level mapping $\ell$ is associated with
the predicates occurring in clauses {\tts 1}--{\tts 9}:
$\ell(${\tts append}$) = \ell(${\tts reverse}$) = 2$, and
$\ell(${\tts snoc}$) = \ell(${\tts len}$) = 1$. \hfill $\Box$

\end{example}

\subsection{Procedure $\mathit{Diff\mbox{-}Define\mbox{-}Fold}$}
\label{subsec:diff}

In order to present the $\mathit{Diff\mbox{-}Define\mbox{-}Fold}$ procedure used by Algorithm~\Diff,  
first we introduce the following notions.

Given a conjunction $G$ of atoms, by $\mathit{bvars}(G)$ 
we denote the set of variables in~$G$ that have a basic type.
Similarly, by $\mathit{adt\mbox{-}vars}(G)$
we denote the set of variables in $G$ that have an ADT type.

\begin{definition}
We say that an atom $($or a clause$)$ {\em has basic types} if {\em all\/} its 
arguments $($or atoms, respectively$)$
have a basic type.
An atom $($or a clause$)$ {\em has ADTs\/} if {\em at least one} of its arguments $($or atoms, respectively$)$ has an ADT type. 
\end{definition}

\begin{definition}
{Given a set $($or a conjunction$)$ $S$ of atoms, $\mathit{SharingBlocks}(S)$ 
denotes the partition of $S$ with respect to
the reflexive, transitive closure, denoted~$\Downarrow_S$, of the relation~$\downarrow_S$ defined as follows.
Given two atoms $A_1$ and $A_2$ in $S$, 
$A_1\! \downarrow_S\! A_2$ holds 
iff $\mathit{adt\mbox{-}vars}(A_1) \cap 
\mathit{adt\mbox{-}vars}(A_2)\!\neq\! \emptyset$.}
The elements of the partition are called the {\em sharing blocks} of~$S$.
We say that $S$  is  
{\em connected} if $\mathit{SharingBlocks}(S)=\{S\}$.
\end{definition}

\begin{definition}
A {\em generalization}  of a pair $(c_1,c_2)$ of constraints is a constraint, denoted $\alpha (c_1,c_2)$,
such that $\mathbb D \models\forall (c_1 \rightarrow \alpha (c_1,c_2))$ and 
$\mathbb D \models\forall (c_2 \rightarrow \alpha (c_1,c_2))$.
\end{definition}

In particular,
we consider the following generalization operator based on 
{\it widening}~\cite{CoH78,Fi&13a}.
Suppose that $c_1$ is the conjunction $(a_1,\ldots,a_m)$ of atomic 
constraints, then $\alpha (c_1,c_2)$ is defined to be
 the conjunction of 
all $a_i$'s in $(a_1,\ldots,a_m)$
such that \mbox{$\mathbb D\!\models\!\forall (c_2\! \rightarrow\! a_i)$}.
In order to improve the efficacy of generalization, when some of the $a_i$'s 
are {\it LIA} equalities, they are split into conjunctions of 
{{\it LIA}} inequalities
before applying widening.

\begin{definition}\label{def:projection}
For any constraint $c$ and tuple $V$ of variables, the {\em projection} 
of~$c$ onto $V$ is a constraint $\pi(c,V)$ such that: 
$(i)$~$\mathit{vars}(\pi(c,V))\!\subseteq\! V$, and
$(ii)$~$\mathbb D \models \forall (c\!\rightarrow\!\pi(c,V))$.
\end{definition}

In our implementation,
$\pi(c,V)$ is computed by applying to the formula $\exists Y.\, c$, where \mbox{$Y\!=\! \vars(c)\! \setminus\! V,$}
a quantifier elimination algorithm for the {theories} of booleans and 
{\it rational} {(not integer)}
numbers. This implementation is safe in our context, 
because it guarantees properties~$(i)$ and~$(ii)$ 
{of Definition~\ref{def:projection}}, 
and avoids relying on
modular arithmetic, as usually done when eliminating quantifiers in 
{\it LIA}~\cite{Rab77}.

\begin{definition}%
For two conjunctions $G_1$ and $G_2$ of atoms, we say that 
$G_1$ {\em atomwise subsumes} $G_2$, denoted $G_1\Embedded G_2$,
if $G_1$ is the conjunction $(A_1,\ldots,A_n)$ 
and there exists a subconjunction $(B_{1},\ldots, B_{n})$
of atoms of 
$G_2$ $($modulo reordering$)$ and substitutions 
$(\vartheta_{1},\ldots,\vartheta_{n})$ such that, 
for $i\!=\!1,\ldots,n,$ we have that $B_{i}\!=\!A_{i}\vartheta_{i}$.
\end{definition}

Now let us present the $\mathit{Diff\mbox{-}Define\mbox{-}Fold}$ procedure 
(see Figure~\ref{fig:Diff}).
At each iteration of the body of the {\bf for} loop, 
the $\mathit{Diff\mbox{-}Define\mbox{-}Fold}$ procedure removes the ADT terms occurring in
a sharing block $B$ of the body of a clause~$C\!:$ $H\!\leftarrow\! c,  B, G'$ of 
$\mathit{InCls}$ (initially, $\mathit{InCls}$ is a set of goals whose head is {\it false}).
This is done by possibly introducing some new definitions using 
the Definition Rule~R1
and applying the Folding Rule~R3. To allow folding, some applications
of the Differential Replacement Rule~R7  may be needed.
We have the following four cases.

\begin{figure}[!ht]
\noindent \hrulefill \nopagebreak

\noindent {\bf Procedure $\mathit{Diff\mbox{-}Define\mbox{-}Fold}(\mathit{InCls},\mathit{Defs}, 
	\mathit{NewDefs},\mathit{FldCls})$}
\\
{\em Input}\/: A set {\it InCls} of clauses and a set {\it Defs} of definitions;
\\
{\em Output}\/: A set {\it NewDefs} of definitions and a set $\mathit{FldCls}$ 
of clauses with basic types.

\vspace{-2mm}
\noindent \rule{2.0cm}{0.2mm}

\noindent $\mathit{NewDefs} := \emptyset; \ \mathit{FldCls}:= \emptyset$;

\noindent {\bf for} each clause $C$: $H\leftarrow c, G$ in $\mathit{InCls}$ {\bf do}

\noindent 
\hspace{3mm}{\bf if} $C$ has basic types {\bf then}
$\mathit{InCls}\! :=\! \mathit{InCls}\!\setminus\! \{C\}$;\ $\mathit{FldCls}:=\mathit{FldCls}\cup\{C\}$

\noindent 
\hspace*{3mm}{\bf else}
\vspace{1mm}

\hspace{6mm}
\vline\hspace{1mm}\begin{minipage}{11.4cm}
let $C$ be $H\leftarrow c, B, G'$ 
where $B$ is a sharing block in $G$ such that $B$ contains at least one atom 
that has ADTs;
\\
$\bullet$ ({\bf Fold}) {\bf if} there is a clause $D$: $\mathit{newp}(U) \If d, 
B'$, which is a variant of a clause\\ 
\hspace*{3.7mm}in  $\mathit{Defs}\cup \mathit{NewDefs}$, with 
$U\!=\!\textit{bvars}(\{d,B'\})$,
such that:
(i)~$B=B'\vartheta$, {for some\\ \hspace*{3.7mm}substitution $\vartheta$ %
{acting on $\mathit{adt\mbox{-}vars}(B')$ only,}}
	and (ii)~$\mathbb D \models\forall (c \rightarrow d)$, 
{\bf then}\\
\hspace*{3.7mm}\underline{\Down fold} $C$ using $D$ and derive $E$: $H\!\leftarrow c, \mathit{newp}(U),\! G'$; \\
$\bullet$ ({\bf Generalize}) {\bf else if} there is a clause 
$D$: $\mathit{newp}(U) \If d, B'$, which is a variant \\ 
\hspace*{3.7mm}of a clause in $\mathit{Defs}\cup
\mathit{NewDefs}$, with 
$U\!=\!\textit{bvars}(\{d,B'\})$,
such that: %
(i)~$B=B'\vartheta$, {\\ \hspace*{3.7mm}for some
	 substitution $\vartheta$ %
     acting on $\mathit{adt\mbox{-}vars}(B')$ only,} and
 (ii)~$\mathbb D \not\models\forall (c \rightarrow d)$,
{\bf then}
\\ %
\hspace*{4.5mm}\vline\hspace{1.5mm}\begin{minipage}{10.8cm} 
\underline{\Down introduce definition} $\mathit{GenD}$: $\mathit{genp}(U) \If \alpha(d,c), B'$
\\
\underline{\Down fold} $C$ using $\mathit{GenD}$ and derive $E$: $H\leftarrow c, \mathit{genp}(U), G'$; 
\\
$\mathit{NewDefs} := \mathit{NewDefs} \cup \{\mathit{GenD}\}$;
\end{minipage}\vspace{.5mm}
\\ 
$\bullet$ ({\bf Diff}-{\bf Introduce}) {\bf else if} there is a clause 
$D$:\,$\mathit{newp}(U) \If d, B'$, which is a\\\hspace*{3.9mm}variant of 
a clause in\,$\mathit{Defs}\,\cup \mathit{NewDefs}$, with 
$U\!=\!\textit{bvars}(\{d,\!B'\})$ and  
$B' \Embedded B$, {\bf then}
\\ %
\hspace*{4.5mm}\vline\hspace{1.5mm}\begin{minipage}{10.8cm} 

take a maximal connected subconjunction $M$ of $B$, if any, such that: 
\hangindent=3mm
\\
(i) $B\!=\!(M, F(X;Y))$, for some \Down non-empty conjunction $F(X;Y)$,
(ii) $B'\vartheta=(M, R(V;W))$, {for some
	 substitution $\vartheta$ %
     acting on $\mathit{adt\mbox{-}vars}(B')$ only}  and
$W \cap \vars(C) \!= \!\emptyset$, and 
(iii)~for every atom~$A$ in $F(X;Y)$, $\ell(H)>\ell(A)$; 

\underline{\Down introduce definition} \hangindent=15mm $\widehat{D}$: $\mathit{diff}(Z) \leftarrow \pi(c,X),F(X;Y),R(V;W)$
 
\hspace*{3mm}where $Z\!=\!\mathit{bvars}(\{F(X;Y),R(V;W)\})$;

$\mathit{NewDefs} := \mathit{NewDefs} \cup \{\widehat{D}\}$;

\underline{\Down replace} $F(X;Y)$ by $(R(V;W), \mathit{diff}(Z))$ in $C$, and derive  clause\\
\hspace*{3mm}\rule{0mm}{2.9mm}$C'$: $H\leftarrow c, M, R(V;W), \mathit{diff}(Z), G'$;

{\bf if} $\mathbb D \models\forall (c \rightarrow d)$
               
\hspace*{3mm}{\bf then}~\underline{\Down fold} $C'$ using $D$ and derive $E$: $H\leftarrow c,\mathit{newp}(U), \mathit{diff}(Z), G'$;

\hspace*{3mm}{\bf else} %
\hspace{.5mm}\vline\hspace{1.5mm}\begin{minipage}[t]{9.5cm}%
\underline{\Down introduce\,definition}\,$\mathit{GenD}$:\,$\mathit{genp}(U) \!\If\!\alpha(d,c), B'$;

\underline{\Down fold} $C'$ using $\mathit{GenD}$ and derive $E$: $H\leftarrow c, \mathit{genp}(U), \mathit{diff}(Z), G'$; 

$\mathit{NewDefs} := \mathit{NewDefs} \cup \{\mathit{GenD}\}$;
\end{minipage} %
\end{minipage}

\noindent
$\bullet$ ({\bf Project}) {\bf else}

\hspace*{4.5mm}\vline\hspace{1.5mm}\begin{minipage}{10.8cm}
\underline{\Down introduce definition} $\mathit{ProjC}$: $\mathit{newp}(U) \If \pi(c,Z), B$ 
\ where $U\!=\!\mathit{bvars}(B)$ \\ 
and $Z$ are the input variables of a basic type in $B$;

\underline{\Down fold} $C$ using $\mathit{ProjC}$ and derive  
clause $E$: $H\leftarrow c, \mathit{newp}(U), G'$; 

$\mathit{NewDefs} := \mathit{NewDefs} \cup \{\mathit{ProjC}\}$;
\end{minipage}

\end{minipage}	%
\vspace{-1mm}

\noindent 
\hspace{3mm}$\mathit{InCls}\! :=\! (\mathit{InCls}\setminus \{C\}) \cup \{E\}$;

\nopagebreak %
\noindent \hrulefill
\vspace*{-1mm}
\caption{The {\it Diff-Define-Fold} procedure.
According to rule R1, the level of every new predicate~(either 
$\mathit{genp}$, or $\mathit{diff}$, or $\mathit{newp}$) 
introduced by the procedure, is equal to the maximum level of 
the atoms occurring in the body of its definition.
\label{fig:Diff}}
\vspace*{-4mm}
\end{figure}	

\smallskip
\noindent
$\bullet$ ({\bf Fold}). In this case the ADT terms in $B$ 
can be removed {by folding using a definition that has already
been introduced.
In particular, let us} suppose that $B$ is an instance, via a substitution $\vartheta$, of the 
conjunction of atoms in the body of a definition $D$ introduced at a previous iteration of the
{\it Diff-Define-Fold} procedure,
and constraint~$c$ in~$C$ entails the constraint in~$D$.
Since we have assumed that all terms of a basic type occurring in an atom are
distinct variables (see Section~\ref{sec:CHCs}), and the variables of $D$ 
can be freely renamed, we require that $\vartheta$ acts on
ADT variables only, and hence it is the identity on the variables of a basic type.
A similar assumption is also made in the next two cases (Generalize) and (Diff-Introduce).
Then, we remove the ADT arguments occurring in $B$ by folding $C$ 
using~$D$.
Indeed, by construction, 
all variables in the head of every definition introduced by 
Algorithm~\Diff~have a basic type. 

\vspace{1.5mm}
\noindent
$\bullet$ ({\bf Generalize}).
Suppose that the previous case does not apply.
Suppose also that there exists a definition $D$, introduced at a previous iteration of the
{\it Diff-Define-Fold} procedure, such that 
the sharing block $B$ is an instance of the conjunction $B'$ of 
the atoms in the body of $D$
and, unlike the (Fold) case, the constraint~$c$~in $C$ does {\em not} 
entail the constraint $d$ in $D$.
We introduce a new definition
$\mathit{GenD}\!:$ $\mathit{genp}(U) \If \alpha(d,c), B'$,
where: (i)~by construction, the constraint $\alpha(d,c)$ is a generalization 
of~$d$ such that $c$ entails $d$,
and (ii)~$U$ is the tuple of the variables of a basic type in $(d,B)$.
Then, we remove the ADT arguments occurring in~$B$ by folding~$C$ 
using~$\mathit{GenD}$. 

\smallskip
\noindent
$\bullet$ ({\bf Diff}-{\bf Introduce}).
Suppose that the previous two cases do not apply 
because the sharing block $B$ in clause~$C$ is not an instance of the conjunction of atoms in the body of any
definition introduced at a previous iteration of the procedure.
Suppose, however, that
$B$ {\em partially matches} the body of {an already introduced}
definition
$D$: $\mathit{newp}(U) \If d, B'$, that is, 
(i)~$B\!=\!(M, F(X;Y))$, and (ii)~for some substitution $\vartheta$ acting on
$\mathit{adt\mbox{-}vars}(B')$ only,
$B'\vartheta\!=\!(M, R(V;W))$ (see Figure~\ref{fig:Diff} for details).
Then, we introduce a difference predicate {\it diff\/} 
{defined by the clause} %
$\widehat{D}$: $\mathit{diff}(Z) \leftarrow \pi(c,X),F(X;Y),R(V;W),$
where $Z\!=\!\mathit{bvars}(\{F(X;Y),$ $ R(V;W)\})$ and, 
by rule~R7, we replace the conjunction
$F(X;Y)$ by the new conjunction $(R(V;W), \mathit{diff}(Z))$ in the body of~$C$,
thereby deriving $C'$. 
Finally, we remove the ADT arguments in~$B$ by folding $C'$ using either $D$ 
(if $c$ entails $d$) or 
a clause $\mathit{GenD}$ whose constraint is the generalization $\alpha(d,c)$
of the constraint~$d$ (if $c$ does {\em not} entail $d$) (again, 
see Figure~\ref{fig:Diff} for details).  

\smallskip
\noindent
$\bullet$ ({\bf Project}). Suppose that none of the previous three 
cases apply. Then,
we first introduce a new definition $\mathit{ProjC}$: 
$\mathit{newp}(U)\! \If\! \pi(c,Z), B$, 
 where $U\!=\!\mathit{bvars}(B)$ and $Z$ are the input variables of basic 
types in $B$, and then we can remove the ADT arguments occurring in the 
sharing block $B$ by {folding} $C$ using $\mathit{ProjC}$.

\begin{example}[Reverse, Continued]\label{ex:rev2}
The body of goal~{\tts 1} (see Section~\ref{sec:IntroExample}) has a single sharing block, that is,

\smallskip
\noindent
$B_1$:~~{\tts append(Xs,Ys,Zs), reverse(Zs,Rs), len(Xs,N0), len(Ys,N1), len(Rs,N2)}  

\smallskip
\noindent
Indeed, we have that {\tts append(Xs,Ys,Zs)} shares a list variable with 
each of atoms {\tts reverse(Zs,Rs)}, {\tts len(Xs,N0)}\!, and {\tts len(Ys,N1)},
and atom {\tts reverse(Zs,Rs)} shares a list variable with {\tts len(Rs,N2)}.
None of the first three cases (Fold), (Generalize), or 
({Diff}-{Introduce}) applies, because $\mathit{Defs}\cup
\mathit{NewDefs}$ is the empty set. 
Thus, Algorithm~\Diff~introduces the following new definition (see also 
Section~\ref{sec:IntroExample}):

\vspace{-2mm}

{\small
\begin{verbatim}
D1. new1(N0,N1,N2) :- append(Xs,Ys,Zs), reverse(Zs,Rs), len(Xs,N0), 
                      len(Ys,N1), len(Rs,N2).
\end{verbatim}
}

\vspace{-2mm}

\noindent
where: 
(i)~{\tts new1} is a new predicate symbol,
(ii)~the body is the sharing block~$B_{1}$,
(iii)~{\tts N0,N1,N2} are the variables of basic types in $B_{1}$,  and
(iv)~the constraint is the empty conjunction {\tts true},
that is, the projection of the constraint {\tts N2=\textbackslash=N0+N1} occurring
in goal {\tts 1}
onto the input variables of basic types in~$B_{1}$ (i.e., the empty set,
as {\tts N0,N1,N2} are all output variables).
In accordance with rule R1, we set 
 $\ell(${\tts new1}$) = \mathit{max}\{\ell(${\tts append}$), \ell(${\tts reverse}$), \ell(${\tts len}$) \} = 2$.

By folding, from goal~{\tts 1} 
we derive a new goal without occurrences 
of list variables:

\smallskip
\noindent
{\tts{10. false :- N2=\textbackslash=N0+N1, new1(N0,N1,N2).}} 

\smallskip
\noindent
The presentation of this example will continue
in Example~\ref{ex:rev-continued} (see
Section~\ref{subsec:proc-unfold-replace}).~\hfill $\Box$
\end{example}

\vspace{-6mm}

\subsection{Procedures $\mathit{Unfold}$ and $\mathit{Replace}$}
\label{subsec:proc-unfold-replace}

The $\mathit{Diff\mbox{-}Define\mbox{-}Fold}$ procedure may
introduce new definitions with ADTs in their bodies (see, for instance, 
clause~{\tts D1} defining predicate {\tts new1} in Example~\ref{ex:rev2}). 
These definitions are added to {\it NewDefs} 
and transformed %
by the $\mathit{Unfold}$ and $\mathit{Replace}$ procedures.

Procedure {\it Unfold} (see Figure~\ref{fig:unfoldProc}) repeatedly 
applies rule~R2 in two phases.
In Phase~1 the procedure unfolds a given clause in {\it NewDefs} with respect to 
so-called {\em source} atoms in its body.
Recalling that each atom is the relational translation of a
function call, the source atoms represent innermost 
function calls in the functional expression corresponding
to the clause body.
The unfolding steps of Phase~1 may determine, by unification,
the instantiation of some input variables.
Then, in Phase~2 these instantiations are taken into account 
for performing further unfolding steps.
Indeed, the procedure selects for unfolding only atoms whose input 
arguments are instances of the corresponding arguments in the heads of their
matching clauses. 	

\begin{figure}[!ht]
\noindent \hrulefill
	
	\noindent {\bf Procedure $\mathit{Unfold}(\mathit{NewDefs},\mathit{Ds},\mathit{UnfCls})$}
	\\
	{\em Input}\/: A set $\mathit{NewDefs}$ of definitions and a set $\mathit{Ds}$ of definite clauses;
	\\
	{\em Output}\/: A set $\mathit{UnfCls}$ of definite clauses.
	
	\vspace*{-2.5mm} \noindent \rule{2.0cm}{0.2mm}

	\noindent
	$\mathit{UnfCls} := \mathit{NewDefs}$; 

	\vspace*{.5mm}
    \noindent
	Phase~1. \vline\vline\hspace{1.5mm}\begin{minipage}[t] {10.7cm}
	\noindent\hangindent=2mm
	- {For each clause $C$ in $\mathit{UnfCls}$, 
	mark as unfoldable a set $S$ of atoms in the body of $C$ such that:
	(i)~there is an atom $A$ in $S$ with $\ell(H)\!=\!\ell(A)$, {where $H$ is the head of $C$,} and
	(ii)~all atoms in $S\setminus \{A\}$ are source atoms such that
	every source variable of the body of $C$ occurs in $S$;}
	
	\noindent\hangindent=2mm
	- {\bf while} there exists a clause $C$: $H \leftarrow c, {L}, A, {R}$\, in $\mathit{UnfCls}$, 
	for some conjunctions $L$ and~$R$ of atoms, such that $A$ is an
	unfoldable atom {\bf do}
	
\vspace*{0.5mm}
	
\noindent
\hspace{5mm}$\mathit{UnfCls}:=(\mathit{UnfCls}\setminus\{C\}) \cup \mathit{Unf(C,A,Ds)}$;
	\end{minipage}

\vspace*{0.5mm}
\hspace{13mm}\rule{2.0cm}{0.2mm}

    \noindent
	Phase~2. \vline\vline\hspace{1.5mm}\begin{minipage}[t] {10.7cm}
	- Mark as unfoldable all atoms in the body of each clause in $\mathit{UnfCls}$;
	
\noindent \hangindent=2mm
- {\bf while} there exists a clause $C$: 	\mbox{$H \leftarrow c, {L}, A, {R}$} 
in $\mathit{UnfCls}$,  for some conjunctions $L$ and~$R$ of atoms,
	such that $A$ is a head-instance %
	with respect to~{\it Ds} 
	and $A$ is either unfoldable or descending {\bf do}
	
\vspace*{0.5mm}
	\noindent
	\hspace{5mm}$\mathit{UnfCls}:=(\mathit{UnfCls}\setminus\{C\}) \cup \mathit{Unf(C,A,Ds)}$;	\end{minipage}
	
\vspace*{0.5mm}	
\noindent \hrulefill
\vspace{-2mm}
\caption{The {\it Unfold} procedure.\label{fig:unfoldProc}}
\vspace*{-3mm}	
\end{figure}

In order to present the {\it Unfold} procedure in a formal way, we need the following notions.

\begin{definition}
A variable $X$ occurring in a conjunction $G$ of atoms is said to be 
a {\em source variable} if it is an input variable for an atom
in $G$ and not an output variable of any atom in $G$.
An atom $A$ in a conjunction $G$ of atoms is said to be a {\em source atom} if
all its input variables are source variables.
\end{definition}

For instance, in clause~{\tts 1} of 
Section~\ref{sec:IntroExample}, where the input variables of the
atoms {\tts append(Xs,Ys,Zs)}, {\tts reverse(Zs,Rs)}, {\tts len(Xs,N0)}, 
{\tts len(Ys,N1)}, and {\tts len(Rs,N2)} 
are {\tts (Xs,Ys)}, {\tts Zs}, {\tts Xs}, {\tts Ys}, and {\tts Rs}, respectively,
there are the following  three source atoms:
{\tts append(Xs\!,Ys\!,Zs)}, {\tts len(Xs\!,N0)}, and 
{\tts len(Ys\!,N1)}. 
These three atoms correspond to the innermost 
function calls which occur in the functional expression 
\mbox{{\tts len(reverse(append xs ys))} {\tts =}\hspace*{-1.6mm}{\tts /} 
{\tts (len xs)\,+\,(len ys)})}
corresponding
to the clause body.

\begin{definition}\label{def:head-instance}
An atom $A(X;Y)$ in the body of clause $C$: \mbox{$H \leftarrow c, {L}, A(X;Y), {R}$}
is a {\em head-instance} with respect to~a set~{$\mathit{Ds}$} of clauses if, 
for every clause $K\leftarrow d, B$ in~{$\mathit{Ds}$} such that: 
$(1)$~there exists a most general unifier $\vartheta$ of $A(X;Y)$ and
$K$, and $(2)$~the constraint $(c, d)\vartheta$ is
satisfiable, we have that $\vartheta$ is a variable renaming for~$X$.
\end{definition}

Thus, $A(X;Y)$ is a head-instance, if for all clause heads $K$ in $\mathit{Ds}$
the input variables $X$ are not instantiated by unification with $K$.
For instance, {with respect to the set $\{${\tts 2},\,{\tts 3}$\}$ of clauses of Section~\ref{sec:IntroExample}}, 
the atom {\tts append([X|Xs],Ys,Zs)} is a head-instance, while the atom
{\tts append(Xs,Ys,Zs)} is~not.

Recall that in a set {\it Cls}  of clauses, predicate $p$ {\it immediately depends on} predicate $q$,
if in\,{\it Cls} there is a clause of the form $p(\ldots) \leftarrow \ldots, q(\ldots), \ldots$
The {\it depends on} relation is the transitive closure of the {\it immediately depends on}
relation~\cite{Apt90}.

\begin{definition}
Let $\prec$ be a well-founded ordering on tuples of terms such that, for
all tuples of terms $t$ and $u,$ if $t\!\prec\! u$, then, for all substitutions $\vartheta,$ 
\mbox{$t\vartheta\!\prec\! u\vartheta$.} A predicate $p$ is {\em descending} with respect to~$\prec$
if, for all clauses, $p(t;u) \leftarrow c,\, p_1(t_1;u_1),\ldots,p_n(t_n;u_n),$
for $i\!=\!1,\ldots,n,$ if $p_i$ depends on $p$, then $t_i\!\prec\! t$.
An atom is descending if its predicate is descending.
\end{definition}

The well-founded ordering $\prec$ we use in our implementation is based on the {\it subterm}
relation and is defined as follows: {$(u_1,\ldots,u_m)\!\prec\! (v_1,\ldots,v_n)$
if for every~$u_i$ there exists $v_j$ such that $u_i$ is a (non necessarily strict) subterm of $v_j$, 
and there exists $u_i$ which is a {strict} subterm
of some~$v_j$.} For instance, the predicates {\tts append}, 
{\tts reverse}, {\tts snoc}, and {\tts len} in
our running example are all descending.

To control the application of rule R2 in Phases~1 and~2 of the {\it Unfold} procedure 
we mark as {\it unfoldable} some atoms in the body of a clause.
If we unfold with respect to~atom~$A$ clause $C$: \mbox{$H \!\leftarrow c, L, A, R$}, then 
the marking of the clauses
in $\mathit{Unf(C,A,Ds)}$ is {done} as follows: 
(i)~each atom derived from $A$ is not marked 
as unfoldable, and (ii)~each atom~$A''$ inherited from an atom~$A'$, 
different from~$A$,
in the body of $C$ is marked as unfoldable iff $A'$ is marked as unfoldable.

In Phase~1, for each clause $C$ in $\mathit{NewDefs}$ 
the procedure marks as unfoldable a non-empty set~$S$ of atoms in
the body of $C$ consisting of: 
(i)~an atom $A$ such that $\ell(H)\!=\!\ell(A)$, 
where $H$ is the head of $C$, and 
(ii)~a set of source atoms (possibly including $A$) such that
every source variable of the body of $C$ occurs in $S$.
Then, the procedure unfolds with respect to all
unfoldable atoms. 
Note that atom~$A$ exists because, by construction,
when we introduce a new predicate during the $\mathit{Diff\mbox{-}Define\mbox{-}Fold}$ procedure,
we set the level of 
the new predicate to the maximal level of an atom in the body of its definition.
The unfolding with respect to~$A$ 
enforces Condition~(U) of Theorem~\ref{thm:unsat-preserv},
and hence the soundness of Algorithm~\Diff.

In Phase~2 the instantiations of input variables determined by the unfolding steps of Phase~1 
are  {taken into account} for further applications of rule~R2.
Indeed, clauses are unfolded with respect to atoms which are head-instances and, 
in particular,
unfolding with respect to head-instances which are descending atoms, is repeated until no such atoms are present.

The termination of the procedure {\it Unfold} 
 is ensured by the following two facts:
(i) if a clause $C$ has $n\, (\geq\!1)$ atoms marked as unfoldable, and 
clause~$C$ is unfolded with respect to an atom~$A$
that is marked as unfoldable, then each clause in $\mathit{Unf(C,A,Ds)}$ 
has $n\!-\!1$ atoms marked as unfoldable, and 
(ii)~since $\prec$ is a well-founded ordering, 
it is not possible to perform an infinite sequence of applications
of the Unfolding Rule R2 with respect to descending atoms.

\begin{example}[Reverse, Continued]\label{ex:rev-continued}
The {\it Unfold} procedure marks as unfoldable atom
{\tts append(Xs,Ys,Zs)} in the body of clause
{\tts D1}, which has the same level as the head of the clause.
Atom {\tts append(Xs,Ys,Zs)} is also a source atom containing all
the input variables of the body of clause {\tts D1} (that is, {\tts Xs}
and {\tts Ys}).
Then, by unfolding clause {\tts D1} with respect to {\tts append(Xs,Ys,Zs)},
we get:

\vspace{1mm}

\noindent
{\tts 11. new1(N0,N1,N2)\!\! :- reverse(Ys,Rs), len([],N0), len(Ys,N1), len(Rs,N2).}

\noindent
{\tts 12. new1(N0,N1,N2)\!\! :- append(Xs,Ys,Zs)\!,\! reverse([X|Zs],Rs)\!,\! len([X|Xs],N0)\!,}

\noindent
\hspace{36mm}{\tts  len(Ys,N1), len(Rs,N2).}

\vspace{1mm}

\noindent
Now, atoms {\tts len([],N0)}, {\tts reverse([X|Zs],Rs)}, and {\tts len([X|Xs],N0)}
are all head-instances, and hence the procedure unfolds clauses {\tts 11} and {\tts 12}
with respect to these atoms. We get:

\vspace*{-2mm}

{\small
	\begin{verbatim}
13. new1(N0,N1,N2) :- N0=0, reverse(Zs,Rs), len(Zs,N1), len(Rs,N2).
14. new1(N01,N1,N21) :- N01=N0+1, append(Xs,Ys,Zs), reverse(Zs,Rs), 
	                        len(Xs,N0), len(Ys,N1), snoc(Rs,X,R1s), 
	                        len(R1s,N21).
	\end{verbatim}
}

\vspace*{-2mm}
The presentation of this transformation will continue in 
Example~\ref{ex:rev-continued2} below. \hfill $\Box$
\end{example}

\noindent
{Procedure {\it Replace} (see Figure \ref{fig:replaceProc}) applies rules~R4, R5, and~R6 as long as possible.
{\it Replace} terminates because each application of one of those rules decreases 
the number of atoms. }

\begin{figure}[!ht]
	\noindent \hrulefill
	
	\noindent {\bf Procedure $\mathit{Replace}(\mathit{UnfCls},\mathit{Ds},\mathit{RCls})$}
	\\
	{\em Input}\/: Two sets $\mathit{UnfCls}$ and $\mathit{Ds}$ of definite clauses;
	\\
	{\em Output}\/: A set $\mathit{RCls}$ of definite clauses.
	
	\vspace*{-2.5mm} \noindent \rule{2.0cm}{0.2mm}

	\noindent
	$\mathit{RCls} := \mathit{UnfCls}$; 

	\vspace*{.5mm}
	\noindent
	\begin{minipage}[t] {11.8cm}
		\noindent\hangindent=4mm
		{\bf repeat} 
		\\
		{\bf if} there is a clause $C\in \mathit{RCls}$ such that rule R4 is applicable to $C$\\ \hspace*{3mm}{\bf then}
		$\mathit{RCls} := \mathit{RCls} \setminus \{C\}$;
		\smallskip		\\
        {\bf if} there is a clause $C\in \mathit{RCls}$ such that the Functionality Rule R5 is applicable to $C$
        with respect to $\mathit{RCls}\cup \mathit{Ds}$, thus deriving a new clause $D$\\ \hspace*{3mm}{\bf then }
        $\mathit{RCls} := (\mathit{RCls} \setminus \{C\}) \cup \{D\}$;
		\smallskip\\
        {\bf if} there is a clause $C\in \mathit{RCls}$ such that the Totality Rule R6 is applicable to $C$
        with respect to $\mathit{RCls}\cup \mathit{Ds}$, thus deriving a new clause $D$\\ \hspace*{3mm}{\bf then }
        $\mathit{RCls} := (\mathit{RCls} \setminus \{C\}) \cup \{D\}$;		
		
		\smallskip
		{\bf until} no rule in $\{$R4, R5, R6$\}$ is applicable to a clause in $\mathit{RCls}$
		
	\end{minipage}
		
	\smallskip	
	\noindent \hrulefill
	\vspace{-2mm}
	\caption{The {\it Replace} procedure.\label{fig:replaceProc}}
	\vspace*{-3mm}	
\end{figure}

\begin{example}[Reverse, Continued]\label{ex:rev-continued2} 
Neither rule~R5 nor rule~R6 is applicable to clauses~{\tts 13} and~{\tts 14}.
Thus, the first iteration of the body of the {\bf while-do}
loop of Algorithm~\Diff~terminates with
$\mathit{InCls}\!=\!\{${\tts 13,14}$\},$
$\mathit{Defs}\!=\!\{${\tts D1}$\},$ and
$\mathit{TransfCls}\!=\!\{${\tts 10}$\}$.

Now, the second iteration starts off by executing the 
{\it Diff-Define-Fold} procedure.
The procedure handles the two clauses {\tts 13} and {\tts 14} in $\mathit{InCls}$.

\smallskip
\noindent
For clause {\tts 13}, the {\it Diff-Define-Fold} procedure
applies case (Project).
Indeed, the body of clause {\tts 13} has the following single sharing block:

\smallskip

\noindent
$B_{13}$:~~{\tts reverse(Zs,Rs), len(Zs,N1), len(Rs,N2)}

\smallskip

\noindent
and there is no clause $\mathit{newp}(V) \If d, B'$ in $\textit{Defs}\,\cup\textit{NewDefs}$
such that $B_{13}$ is an instance of $B'$.
Thus, the procedure 
adds to \textit{NewDef} the following new definition
(see also Section~\ref{sec:IntroExample}):

\smallskip

\noindent
{\tts D2. new2(N1,N2) :- reverse(Zs,Rs), len(Zs,N1), len(Rs,N2).}

\smallskip
\noindent
and, by folding clause~{\tts 13}, we get:

\smallskip

\noindent
{\tts 15. new1(N0,N1,N2) :- N0=0, new2(N1,N2).}

\smallskip

\noindent
which has basic types and hence it is added to {\it FldCls}.
This clause is then added to the output set {\it TransfCls} 
(see Figure~\ref{fig:AlgoR}).

\medskip
\noindent
For clause {\tts 14}, the {\it Diff-Define-Fold} procedure
applies case (Diff-Introduce).
Indeed, the body of clause {\tts 14} has the following single sharing block:

\smallskip

\noindent
$B_{14}$:~~{\tts append(Xs,Ys,Zs)\!,} {\tts reverse(Zs,Rs)\!,}\ {\tts len(Xs,N0)\!,}\ {\tts len(Ys,N1)\!,}\,

\hspace{3.7mm}{\tts snoc(Rs,X,R1s)\!,}\ {\tts len(R1s,N21)}
          
\smallskip

\noindent
and we have that $B_1 \Embedded B_{14},$ where $B_1$ is the body of 
clause {\tts D1},
which is the definition introduced as explained in Example~\ref{ex:rev2} above.
The procedure constructs the conjunctions defined at Points (i)--(iii) of (Diff-Introduce)
as follows:

\smallskip
\noindent
$M$= ({\tts append(Xs,Ys,Zs)\!,} {\tts reverse(Zs,Rs)\!, len(Xs,N0)\!, len(Ys,N1)}),

\noindent
$F(X;Y)$\! =\! ({\tts  snoc(Rs,X,R1s)\!,\;len(R1s,N21)}), ~where 
$X$={\tts  (Rs,X)}, $Y$={\tts  (R1s,N21)},

\noindent
$R(V;W)$\! =\! {\tts len(Rs,N2)}, ~where $V$=\,{\tts (Rs)},  $W$=\,{\tts (N2)}.

\smallskip
\noindent
In this example, $\vartheta$ is the identity substitution. Morevover,
the condition on the level mapping~$\ell$ required in the 
{\it Diff-Define-Fold} Procedure of Figure~\ref{fig:Diff} 
is fulfilled because
$\ell(${\small{\tts new1}}$)\!>\!\ell(${\small{\tts snoc}}$)$ and 
$\ell(${\small{\tts new1}}$)\!>\!\ell(${\small{\tts len}}$)$.
Thus, the definition $\widehat{D}$ 
to be introduced~is:

\vspace{-1mm}

{\small
	\begin{verbatim}
	D3. diff(X,N2,N21) :- snoc(Rs,X,R1s), len(R1s,N21), len(Rs,N2).
	\end{verbatim}
}
\vspace{-1mm}

\noindent
Indeed, we have that: (i)~the projection $\pi(c,X)$ is 
$\pi(${\small{\tts N01=N0+1}},\,{\small{\tts (Rs,X)}}$)$, that is, the empty conjunction {\tts true},
(ii)~$F(X;Y),\, R(V;W)$ is the body of clause~{\tts D3},
and (iii)~the head variables {\tts N2}, {\tts X}, 
and {\tts N21} are the integer
variables in that body. 

Note that:
(i)~clause~{\tts D3} %
is the one we have presented  
in Section~\ref{sec:IntroExample}, 
and (ii)~the relationship between the sharing blocks~$B_{1}$ and~$B_{14}$,
which occur in the body of clauses~{\tts D1} and~{\tts 14}, respectively,
formalizes the notion of mismatch between the bodies of clauses~{\tts D1} 
and~{\tts D1$^{\textstyle *}$} described in Section~\ref{sec:IntroExample},
because clause~{\tts 14} is the same as clause~{{\tts D1$^{\textstyle *}$}.}

Then,  by applying rule~R7 to clause~{\tts 14}, the 
conjunction `{\tts  snoc(Rs,X,R1s)\!,} {\tts len(R1s\!,N21)\!}' can be
replaced by the new conjunction\,`{\tts len(Rs\!,N2)\!,} 
{\tts diff\!(X\!,N2\!,N21)\!}', and we get the clause:

\vspace{-1mm}

{\small
	\begin{verbatim}
	16. new1(N01,N1,N21) :- N01=N0+1, append(Xs,Ys,Zs), reverse(Zs,Rs), 
	                        len(Xs,N0), len(Ys,N1), len(Rs,N2), 
	                        diff(X,N2,N21).
	\end{verbatim}
}

\vspace{-1mm}

\noindent
Finally, by folding clause~{\tts 16} using clause~{\tts D1}, 
we get the following clause:

\vspace{-1mm}

{\small
	\begin{verbatim}
	17. new1(N01,N1,N21) :- N01=N0+1, new1(N0,N1,N2), diff(X,N2,N21).
	\end{verbatim}
}

\vspace{-1mm}

\noindent
which has no list arguments and hence it is added to {\it FldCls}.
This clause is then added to the output set {\it TransfCls}.

\smallskip
Algorithm \Diff~proceeds by applying the \textit{Unfold} and \textit{Replace} procedures
to clauses~{\tts D2} and~{\tts D3}. Then, a final execution of the {\it Diff-Define-Fold}
procedure allows us to fold all clauses with ADT terms 
and derive clauses with basic types,
without introducing any new definition. Thus, \Diff~terminates and its output \textit{TransfCls}
is equal  (modulo variable renaming) to the set \textit{TransfRevCls} of clauses listed in Section~\ref{sec:IntroExample}.
\hfill$\Box$
\end{example}

\subsection{Termination of {Algorithm}~\Diff}

As discussed above, each execution of the {\it Diff-Define-Fold}, 
{\it Unfold}, and {\it Replace} procedures terminates.
However, Algorithm~\Diff~might not terminate because 
new predicates may be introduced
by {\it Diff-Define-Fold} at each iteration of 
the {\bf while}-{\bf do} of~\Diff, and the loop-exit condition
$\textit{InCls}\neq \emptyset$ might be never satisfied.

Thus, {Algorithm}~\Diff~terminates if and only if,
during its execution, the {\it Diff-Define-Fold} procedure 
introduces a finite set of new predicate definitions.
A  way of achieving this finiteness property
is to combine the use of a generalization operator for constraints (see Section~\ref{subsec:diff})
with a suitable generalization strategy for the conjunctions of atoms 
that can appear in the body of the definitions
(see, for instance, the {\em most specific generalization}
used by {\em conjunctive partial deduction}~\cite{De&99}).
It should be noticed, however, that an effect of a badly designed generalization
strategy could be an ADT removal algorithm that often terminates and returns a
set of unsatisfiable CHCs whereas the initial clauses were satisfiable
(in other terms, the transformation would often generate {\em spurious counterexamples}).
  
The study of suitable generalization strategies and also
the study of classes of CHCs for which a suitable modification of
Algorithm~\Diff~terminates are beyond the scope of the present 
paper. 
Instead, in Section~\ref{sec:Experiments}, we evaluate the effectiveness of Algorithm~\Diff~from
an experimental viewpoint.

\subsection{Soundness of {Algorithm}~\Diff}
\label{subsec:soundR}

The soundness of~\Diff~follows from the soundness of the transformation 
rules, and hence we have the following result.

\begin{theorem}[Soundness of {Algorithm}~\Diff] 
\label{thm:soundness-AlgorithmR}
Suppose that {Algorithm}~\Diff~terminates 
for an input set $\mathit{Cls}$ of clauses, and let $\mathit{TransfCls}$
be the output set of clauses.
Then, every clause in $\mathit{TransfCls}$ has basic types, and
if $\mathit{TransfCls}$ is satisfiable, then $\mathit{Cls}$ is satisfiable.
\end{theorem}

\begin{proof}
Each procedure used in {Algorithm}~\Diff~consists of a sequence of 
applications of rules R1--R7.
Moreover, the Unfold procedure ensures that
each clause $H \leftarrow c,B$  introduced by rule R1 is unfolded with respect 
to an atom $A$ in $B$ such that $\ell(H)\!=\!\ell(A)$. 
Thus, Condition~(U) of the hypothesis of Theorem~\ref{thm:unsat-preserv} holds
for any transformation sequence generated by {Algorithm}~\Diff, and hence
the thesis follows from Theorem~\ref{thm:unsat-preserv}. 
\hfill~$\Box$
\end{proof}

\section{Preserving Completeness}
\label{sec:Completeness}
In the previous Sections~\ref{sec:TransfRules} and~\ref{sec:Strategy}, 
we have shown the soundness of the transformation rules, and hence the soundness of 
Algorithm~\Diff.

However, the use of rules R1--R7, with the restrictions mentioned in
Theorem~\ref{thm:unsat-preserv}, does {\em not} preserve completeness, in the sense that, 
we may construct a transformation sequence $P_0 \Rightarrow P_1 \Rightarrow \ldots \Rightarrow P_n$, 
where $P_0$  is a 
satisfiable set of clauses and $P_n$ is a set of unsatisfiable clauses.
Thus, the hypotheses of Theorem~\ref{thm:unsat-preserv} do not guarantee that
 {Algorithm}~\Diff~preserves completeness.
In this section we will introduce some sufficient conditions that guarantee the 
preservation of completeness.

\subsection{Completeness of the transformation rules}
\label{subsec:completeness}

Completeness may be affected by the 
use of rule~R3 or rule~R7, as shown by the following two  examples.
In these examples the variables 
range over the integers~${\mathbb Z}$
or the lists  $\mathbb{L}$ of integers, according to their type.

\begin{example}\label{ex:false-pos-fold}
Let us consider the following set of clauses:

\vspace{1mm}
\noindent
\begin{minipage}[t]{8mm}
$P_{0}$:
\end{minipage}
\begin{minipage}[t]{50mm}
{\small
\begin{verbatim}
1. false :- Y>0, a([],Y).         
2. a([],Y) :- Y=0.       
3. a([H|T],Y) :- Y=1.    
\end{verbatim}
}
\end{minipage}

\smallskip

\noindent
We introduce the following clause defining a new predicate by rule R1:

\vspace{1mm}
\hspace{4mm}{\tts{4. newp(Z) :- a(X,Z).}} 
\vspace{1mm}

\noindent
and we get $P_1\!=\!\{ ${\tts 1,2,3,4}\}. Now, we unfold clause {\tts 4} 
and we derive the clauses:

\vspace{1mm}
\hspace{4mm}{\tts{5. newp(Y) :- Y=0.}}

\hspace{4mm}{\tts{6. newp(Y) :- Y=1. }} 
\vspace{1mm}

\noindent
We get $P_2\!=\!\{${\tts 1,2,3,5,6}\}. Finally, we fold clause {\tts 1} using clause {\tts 4},
which belongs to $\textit{Defs}_2$, and we derive the set of clauses:

\vspace{1mm}
\noindent
\begin{minipage}[t]{8mm}
$P_{3}$:
\end{minipage}
\begin{minipage}[t]{50mm}
{\small
\begin{verbatim}
1f. false :- Y>0, newp(Y).       
2.  a([],Y) :- Y=0.       
3.  a([H|T],Y) :- Y=1. 
5.  newp(Y) :- Y=0.
6.  newp(Y) :- Y=1.    
\end{verbatim}
}
\end{minipage}
\vspace{1mm}

\noindent
Now, we have that $P_{0}$ is satisfiable (because {\tts a([],Y)} holds for
{\tts Y=0} only), while $P_{3}$ is unsatisfiable (because {\tts newp(Y)} holds for {\tts Y=1}). 

Let us explain why, in this example, folding affects completeness.
By applying the Folding Rule~R3 to clause~{\tts 1},
we have replaced atom {\tts a([],Y)}
by atom {\tts newp(Y)}, which, by clause~{\tts 4},
is equivalent to
{\tts $\mathtt{\exists}$X.a(X,Y)} (because~{\tts X} does not occur in the 
head of clause~{\tts 4}).
Thus, folding is based on a substitution for the existentially quantified variable
{\tts X} which is not the identity.
 Now in
$M(P_{1})$, which is \{{\tts a([],0)}, {\tts a([h|t],1)}, 
{\tts newp(0)}, {\tts newp(1)}
$\mid$  {\tts h}\,$\in\!{\mathbb Z}$,  {\tts t}\,$\in\!{\mathbb L}$\},
atoms {\tts a([],Y)} and {\tts newp(Y)} are {\it not equivalent}.
Indeed, $M(P_{1})\models$ {\tts a([],Y)} $\rightarrow$ {\tts newp(Y)}, 
while $M(P_{1})\not\models$ {\tts newp(1)} 
$\rightarrow$ {\tts a([],1)}.
\hfill $\Box$
\end{example}

\begin{example}\label{ex:false-pos-diff}
Let us consider the following set of clauses:

\vspace{1mm}
\noindent
\begin{minipage}[t]{8mm}
$P_{0}$:
\end{minipage}
\begin{minipage}[t]{50mm}
{\small
\begin{verbatim}
1. false :- Y>0, a(X), f(X,Y). 
2. a([]).       
3. f([],Y) :- Y=0.       
4. f([H|T],Y) :- Y=1.    
5. r(X,W) :- W=1.
\end{verbatim}
}
\end{minipage}

\smallskip
\noindent
We introduce the following clause defining the new predicate {\tts diff} by 
rule R1:

\vspace{1mm}
\hspace{4mm}{\tts{6. diff(W,Y) :- f(X,Y), r(X,W). }} 
\vspace{1mm}

\noindent 
where: (i)~{\tts f(X,Y)} is a total, functional atom from 
{\small{\tt{X}}} to {\small{\tt{Y}}}, and 
(ii)~{\small{\tt{r(X,W)}}}  is a total, functional atom from 
{\small{\tt{X}}} to {\small{\tt{W}}}. Thus, we get 
$P_{1}\!=\!\{${\tts{1,2,3,4,5,6}}$\}$ and $\mathit{Defs}_{1}\!=\!\{${\tts{6}}$\}$.
By applying the Differential Replacement Rule~R7, from $P_{1}$ we  derive the following set of clauses:

\vspace{1mm}
\noindent
\begin{minipage}[t]{8mm}
$P_{2}$:
\end{minipage}
\begin{minipage}[t]{90mm}
{\small
\begin{verbatim}
1r. false :- Y>0, a(X), r(X,W), diff(W,Y). 
2.  a([]).        
3.  f([],Y) :- Y=0.       
4.  f([H|T],Y) :- Y=1.    
5.  r(X,W) :- W=1.
6.  diff(W,Y) :- f(X,Y), r(X,W). 
\end{verbatim}
}
\end{minipage}
\smallskip

\noindent 
Now, we have that $P_{0}$ is satisfiable 
(because {\tts a(X)} holds for {\tts X=[]} only, and {\tts f([],Y)} 
holds for {\tts Y=0} only),
while $P_{2}$ is unsatisfiable
(because the body of clause~{\tts 1r} holds for {\tts X=[]} and {\tts W=Y=1}). 

Let us now explain why, in this example, the application of rule~R7 
affects completeness. 
By applying rule~R7, we have replaced atom {\tts f(X,Y)}
by the conjunction `{\tts r(X,W),\;diff(W,Y)}', 
which by clause~{\tts 6} and the totality of {\tts r(X,W)} 
from~{\tts X} to {\tts W}, is implied by {\tts f(X,Y)}.
However, in 
$M(P_{1})$, which is \{{\tts a([])}, {\tts f([],0)}, {\tts f([h|t],1)}, 
{\tts r(u,1)}, {\tts diff(1,0)}, {\tts diff(1,1)} $\mid$ 
{\tts h}\,$\in\!\!{\mathbb Z}$,
{\tts t},\,{\tts u\,}$\in\!\!{\mathbb L}$\}, atom
{\tts f(X,Y)} and the conjunction `{\tts r(X,W),\;diff(W,Y)}' are 
{\it not equivalent}. Indeed, we have that
$M(P_{1}) \not\models$ ({\tts r([],1)}$\,\wedge\,${\tts diff(1,1)})  
$\rightarrow$ {\tts f([],1)}. 

\indent
{In particular, note that when applying rule~R7, 
we have replaced 
{\tts f(X,Y)}, which is functional from~{\tts X} to~{\tts Y}, by 
`{\tts r(X,W),\;diff(W,Y)}', which is not functional from 
{\tts X} to~{\tts (W,Y)} (indeed,  for all {\tts u}$\,\in\!{\mathbb L}$, 
we have that both `{\tts r(u,1),\;diff(1,0)}' and `{\tts r(u,1),\;diff(1,1)}'
do hold).}
This is due to the fact that {\tts diff(W,Y)} is {\em not functional\/} from {\tts W} to {\tts Y}.
\hfill $\Box$
\end{example}

Now, we will give some sufficient conditions that
guarantee that the transformation rules presented in Section~\ref{sec:TransfRules}
are complete, in the sense that
if a  set $P_0$ of CHCs is transformed into a new set $P_n$ by $n$ 
applications of the rules and $P_0$ is satisfiable,
then also $P_n$ is satisfiable.
Thus, when those conditions hold, the converse of the Soundness 
Theorem~\ref{thm:unsat-preserv} holds.

We consider the following Conditions~(E) and~(F) on 
the application of the Folding Rule~R3~\cite{EtG96,TaS84} 
and the Differential Replacement Rule~R7, respectively.

\begin{definition}[Condition E] \label{def:condR}
{\rm Let us assume that: (i)~we apply the Folding Rule~R3 for folding clause $C$: $H\leftarrow c, G_L,Q,G_R$ in $P_i$
using the definition~$D$: $K \leftarrow d, B$, and (ii)~$\vartheta$ is  
a substitution such that~\mbox{$Q\!=\! B\vartheta $} and
$\mathbb D\models \forall(c \rightarrow d\vartheta)$. 
We say that this application of rule~R3 \textit{fulfills Condition} (E) 
if the following  holds:

\vspace{1mm}\noindent
(E)~for every variable \( X\!\in\!\textit{vars}(\{d,B\})\!\setminus\!\textit{vars}(K)\), 

\noindent
\hspace{6mm}E1.~\(
X\vartheta \) is a variable not occurring in \( \{H,c,G_{L},G_{R}\}
\)~ and 

\noindent
\hspace{6mm}E2.~$X\vartheta$ does not occur in the term 
$Y\vartheta$, for any variable~$Y$ occurring\newline
\hspace*{13mm}in $(d, B)$ and different from $X$.
}
\end{definition}

In Condition~(F) below, we consider the particular case, which
is of our interest in this paper,
when rule~R7 is applied within the \textit{Diff-Define-Fold} procedure.

\begin{definition}[Condition F]\label{def:condF}{\rm
Let us assume that we apply the Differential Replacement
Rule R7 to clause $C$: $H\leftarrow c,  G_L, F(X;Y), G_R$ in~$P_i$ using 
the definition $\widehat{D}$: $\mathit{diff}(T_b,W_b,Y_b) \leftarrow d, F(X;Y), R(V;W)$
in $\mathit{Defs}_i$, where $T_b\!=\!\textit{bvars}(X\!\cup V)$,
$W_b\!=\!\textit{bvars}(W)$, and $Y_b\!=\!\textit{bvars}(Y)$.}
{\rm{We say that this application of rule~R7 \textit{fulfills Condition} (F) 
if the following holds:

\smallskip\noindent
\makebox[13mm][r]{(F)~F1.~}atom $\mathit{diff}(T_b,W_b,Y_b)$ is functional 
from $(T_b,W_b)$ to $Y_b$ with respect

\noindent\hspace{13mm}to $\textit{Definite}(P_0) \cup \mathit{Defs}_i$,

\noindent
\makebox[13mm][r]{F2.~}$Y \cap (V \cup \textit{vars}(d)) = \emptyset$, and

\noindent
\makebox[13mm][r]{F3.~}$\mathit{adt}\mbox{-}\mathit{vars}(Y) \cap 
\mathit{adt}\mbox{-}\mathit{vars}(\{H,c, G_L, G_R\}) = \emptyset$.

}}
\end{definition}

The following theorem
guarantees that, if Conditions~(E) and~(F) hold, then
the transformation rules R1--R7 are complete. 

\begin{theorem}[Completeness of the Transformation Rules]
\label{thm:sat-preserv}
Let $P_0 \Rightarrow P_1 \Rightarrow \ldots \Rightarrow P_n$ be a transformation sequence
using rules~{\rm R1--R7}.
Suppose that, for every application of~\,{\rm R3}, Condition {\rm (E)} holds,
and for every application of~\,{\rm \/R7}, Condition {\rm (F)} holds. 
If $P_0$ is satisfiable, then $P_n$
is satisfiable.
\end{theorem}

Note that the applications of R3 and R7 in Examples~\ref{ex:false-pos-fold} 
and~\ref{ex:false-pos-diff} violate Conditions~(E) and~(F),
respectively, and these facts explain why they affect completeness.
Indeed, in Example~\ref{ex:false-pos-fold}, atom {\tts a([],Y)} in 
the body of clause~{\tts 1} is an instance of the body of clause~{\tts4}
via the substitution $\vartheta = \mbox{\tts {\{X/[],Z/Y\}}}$, and $\vartheta$
does not satisfy Condition~(E1).
In Example~\ref{ex:false-pos-diff}, as mentioned above, 
{\tts diff(W,Y)} is not functional from {\tts W} to {\tts Y}, and hence
Condition~(F1) is not fulfilled.

For the proof of Theorem~\ref{thm:sat-preserv} we need some preliminary results.
First we prove the following theorem, which is the converse of Theorem~\ref{thm:unsat}
and is a consequence of Theorem~\ref{thm:pc} reported in Section~\ref{subsec:soundness}.

\begin{theorem}
\label{thm:sat} Let $P_0 \repl \ldots \repl P_n$
be a transformation sequence constructed using rules
{\rm R1} $($Definition$)$, {\rm R2} $($\!Unfolding$)$, {\rm R3} $($Folding$)$, 
and 
\,{\rm R8} $($\!Goal Replacement$)$. Suppose that, for all applications of\,
{\rm R3},
Condition {\rm (E)} 
holds and all goal replacements are 
body strengthenings $($that is, they are applications of rule~{\rm R8} for which 
Condition~$(S)$ of Section~{\rm{\ref{subsec:soundness}}} holds$)$.
If $P_0$ is satisfiable, then $P_n$ is satisfiable.
\end{theorem}

\begin{proof}
As shown in the proof of Theorem~\ref{thm:unsat}, $P_0$ is satisfiable iff 
$P_0 \cup \textit{Defs}_n$ is satisfiable. 
Also in this proof, we consider the transformation sequence  $P'_0\Rightarrow \ldots \Rightarrow P'_n$
obtained from the  sequence $P_0\Rightarrow \ldots \Rightarrow P_n$
by replacing each occurrence of \textit{false} in the head of
a clause by a new predicate symbol $f$.
$P'_0,\ldots,P'_n$ are sets of definite clauses, and thus for $i=0,\ldots,n,$ $\textit{Definite}(P'_i)= P'_i$.
The sequence $P'_0\Rightarrow \ldots \Rightarrow~P'_n$ satisfies the 
hypotheses of Theorem~\ref{thm:pc}, and hence 
\(M(P'_0\cup \mathit{Defs}_n)\supseteq M(P'_n) \). Thus, we have that: 

\noindent
\makebox[12mm][l]{}$P_0$ is satisfiable

\noindent
{implies $P_0 \cup \textit{Defs}_n$ is satisfiable}

\noindent
implies $P'_0 \cup \textit{Defs}_n\cup \{\neg f\}$ is satisfiable

\noindent
{implies $f\not\in M(P'_0 \cup \textit{Defs}_n)$}

\noindent
implies, by Theorem~\ref{thm:pc}, $f\not\in M(P'_n)$

\noindent
\makebox[49mm][l]{implies $P'_n\cup \{\neg f\}$ is satisfiable}

\noindent
implies $P_n$ is satisfiable. \hfill$\Box$
\end{proof}

Now, in order to prove Theorem~\ref{thm:sat-preserv} of 
Section~\ref{sec:TransfRules},  we show that rules~R4--R7 are all body strengthenings.

\medskip

\noindent
Rule R4 (Clause Deletion) is a body strengthening, as 

\smallskip
$M(\textit{Definite}(P_0)\cup \mathit{Defs}_i) \models \forall\, (  {\it false} \rightarrow c \wedge G)$

\smallskip
\noindent
trivially holds.

\medskip

Now let us consider rule~R5 (Functionality). 
Let $F(X,Y)$ be a conjunction of atoms 
that defines a functional relation from~$X$ to~$Y$.
When rule R5 is applied whereby the conjunction $F(X,Y),\ F(X,Z)$
is replaced by the conjunction $Y\!=\!Z,$ $F(X,Y)$, it is the case that 

\smallskip
$M(\mathit{Definite}(P_0)\cup \mathit{Defs}_i) \models \forall\,(Y\!=\!Z \wedge F(X,Y) \rightarrow F(X,Y) \wedge F(X,Z))$

\smallskip
\noindent
Hence, Condition (S) holds and rule~R5 is a body strengthening. 

\medskip

An application of rule R6 (Totality) replaces a conjunction $F(X,Y)$
by {\it true} (that is, the empty conjunction). 
When rule~R6 is applied, it is the case that, by Property~(\textit{Total\/}) of Section~\ref{sec:CHCs},

\smallskip
$M(\mathit{Definite}(P_0)\cup \mathit{Defs}_i) \models \forall\,({\it true} \rightarrow \exists Y.\, F(X,Y))$

\smallskip
\noindent
Hence, rule~R6 is a body strengthening.

\medskip

Now we prove that, when Condition (F) of Definition~\ref{def:condF} holds, rule~R7 is a body strengthening.

\begin{lemma}\label{lemma:R7sat}

Let us consider the following clauses $C$, $\widehat{D}$, and $E$ used when applying
rule~{\rm R7}\,$:$

\smallskip

$C$: $H\leftarrow c,\,  G_L,\, F(X;Y),\, G_R$

$\widehat{D}$: $\mathit{diff}(T_b,W_b,Y_b) \leftarrow d,\, F(X;Y),\, R(V;W)$ 

$E$: $H\leftarrow c,\,  G_L,\, R(V;W),\, \mathit{diff}(T_b,W_b,Y_b),\, G_R$

\smallskip

\noindent
where $Y=(Y_a,Y_b)$, $Y_a=\mathit{adt}\mbox{-}\mathit{vars}(Y)$, and $Y_b=\textit{bvars}(Y)$.
Let us assume that Conditions~{\rm (F1)} and~{\rm (F2)} 
of Definition~$\ref{def:condF}$ hold.
Then, 

\smallskip
\noindent
~~$M(\!\textit{Definite}(P_0) \cup \mathit{Defs}_i)\models \forall\,(c\, \wedge R(V;\!W) 
\wedge \mathit{diff}(T_b,\!W_b,\!Y_b)\rightarrow \exists Y_a.\, c\, \wedge\, F(X;\!Y)).$
\end{lemma}

\begin{proof} %

Let $\mathcal M$ denote $M(\textit{Definite}(P_0) \cup \mathit{Defs}_i)$.
Let~$Y'=(Y_a,Y_b')$, where $Y_b'$ is obtained by renaming the
variables in $Y_b$ with new variables of basic type.
By the totality of $F(X;Y')$, we have: 

\smallskip

$\mathcal M\models \forall\,(c \wedge R(V;W) \wedge \mathit{diff}(T_b,W_b,Y_b)\rightarrow  \exists Y'.\, F(X;Y'))$

\smallskip
\noindent
By the definition of rule R7, $\mathbb{D}\models\forall (c\rightarrow d)$ holds, and we get:

\smallskip

$\mathcal M\models \forall\,(c \wedge R(V;W) \wedge \mathit{diff}(T_b,W_b,Y_b)\rightarrow  \exists Y'.\, d \wedge F(X;Y')\wedge R(V;W))$

\smallskip
\noindent
Now, we have that $Y \cap (X \cup V\cup W \cup \vars(d)) = \emptyset$.
Indeed, (i)~by the definition of rule R7, $W \cap \textit{vars}(C) =\emptyset$,
(ii)~by the notation `$F(X;Y)$', we have that $\textit{vars}(X) \cap \textit{vars} (Y) =\emptyset $,
and (iii) by Condition (F2), $Y \cap (V\cup \vars(d)) = \emptyset$. Then,
${d} \wedge F(X;Y')  \wedge R(V;W)$ is a variant of the body of clause~$\widehat{D}$, 
and since $Y'=(Y_a,Y_b')$, we get:  

\smallskip
$\mathcal M\models \forall\,(c \wedge R(V;W) \wedge \mathit{diff}(T_b,W_b,Y_b)  $

\hspace{2cm}$\rightarrow \exists Y_a,Y'_b.\, F(X;(Y_a,Y_b'))\wedge \mathit{diff}(T_b,W_b,Y'_b))$

\smallskip

\noindent
By Condition (F1), $\mathit{diff}(T_b,W_b,Y_b)$ is functional from $(T_b,W_b)$ to $Y_b$, and we have:

\smallskip

$\mathcal M\models \forall\,(c \wedge R(V;W) \wedge \mathit{diff}(T_b,W_b,Y_b)\rightarrow  \exists Y_a,Y'_b.\, F(X;(Y_a,Y_b'))\wedge  Y_b\!=\!Y'_b)$

\smallskip

\noindent
Thus,

\smallskip
$\mathcal M\models \forall\,(c \wedge R(V;W) \wedge \mathit{diff}(T_b,W_b,Y_b)\rightarrow  \exists Y_a.\, F(X;(Y_a,Y_b)))$

\smallskip
\noindent
and, observing that $Y_a\cap \textit{vars}(c)=\emptyset$, we get the thesis. \hfill $\Box$ \end{proof}

Now, in order to show that an application of rule R7 according to the hypotheses 
of Lemma~\ref{lemma:R7sat} (see also Definition~\ref{def:condF}) is an instance of 
a body strengthening, where in 
clauses~$C$ and~$D$ of rule R8 we consider $c\!=\!\mathit{true}$ 
and $c_{1}\!=\!c_{2}\!=\!c$, we have to show:

\smallskip
\makebox[12mm][c]{(S$_{c}$)}  $M(\textit{Definite}(P_0)\cup \mathit{Defs}_i) \models \forall \, (c \wedge {G}_{2}
\rightarrow \exists T_1.\, c \wedge {G}_{1})$ 

\smallskip
\noindent
where:  

$T_{1} = \mathit{vars}(c \wedge F(X;\!(Y_{a},Y_{b}))) \setminus \mathit{vars}(\{{H}, \mathit{true}, 
{G}_{L},{G}_{R}\})$

$G_{1}=F(X;(Y_{a},Y_{b}))$, and

$G_{2}= R(V;W),\, \mathit{diff}(T_b,W_b,Y_b)$.

\smallskip
\noindent
Now, by Lemma~\ref{lemma:R7sat}, we have:

\smallskip
\makebox[12mm][c]{(L)}~~$M(\!\textit{Definite}(P_0) \cup \mathit{Defs}_i)\models 
\forall\,(c\, \wedge G_{2} \rightarrow \exists Y_a.\, c\, \wedge\,
 F(X;\!(Y_{a},Y_{b})))$ 

\smallskip
\noindent
Since:
(i)~by Condition~(F3) and the fact that the variables of $c$ are all of basic type, we have that:
$Y_a \subseteq T_{1}\!= \!\mathit{vars\/}(c \cup X \cup Y_{a} \cup Y_{b}) 
\setminus \mathit{vars\/}(\{H,\mathit{true},G_L,G_R\})$, and 
(ii)~the variables in $T_{1}\setminus Y_{a}$ are universal quantified in~(L),
we have that~(L) implies~(S$_{c}$). This completes the proof that if Condition (F) 
of Definition~\ref{def:condF} holds, then rule R7 is a body strengthening.

Thus, by taking into account also the facts we have proved above about rules~R4, R5, and R6,
we have the following lemma. 

\begin{lemma}\label{lemma:strengthen}
All the applications of rules {\rm R4, R5, R6}, and the applications of rule~{\rm R7}
where Condition~{\rm (F)} holds $($see Definition~$\ref{def:condF})$, are body strengthenings.
\end{lemma}

\medskip

Now we can present the proof of Theorem~\ref{thm:sat-preserv}.

\smallskip
\noindent
{\it Proof of Theorem~$\ref{thm:sat-preserv}$}.
Let $P_0 \repl \ldots \repl P_n$ 
be a transformation sequence using rules~{\rm{R1}}--{\rm{R7}}.
Suppose that, 
for every application of {\rm R3}, Condition {\rm (E)} holds,
and for every application of {\rm R7}, Condition {\rm (F)} holds. 
Thus, $P_0 \repl \ldots \repl P_n$ can also be constructed by 
applications of rules R1--R3 and applications of rule~R8 which,
by Lemma~\ref{lemma:strengthen}, are all
body strengthenings.
Then, the thesis follows from Theorem~\ref{thm:sat}.
\hfill $\Box$

\subsection{Completeness of Algorithm~\Diff}
\label{subsec:completeAlgo}

We have the following straightforward consequence of Theorem~\ref{thm:sat-preserv}.

\begin{theorem}[Completeness of {Algorithm}~\Diff] 
\label{thm:completeness-AlgorithmR}
Suppose that {Algorithm}~\Diff\ terminates 
for the input set $\mathit{Cls}$ of clauses, and let $\mathit{TransfCls}$
be the output set of clauses.
Suppose also that all applications of rules {\rm R3} and \,{\rm R7} during the execution of Algorithm~\,\Diff~fulfill Conditions~{\rm (E)} and~{\rm (F)}, respectively.
If $\mathit{Cls}$ is satisfiable, then $\mathit{TransfCls}$ is satisfiable.
\end{theorem}

In practice, having constructed the transformation sequence 
$P_0 \Rightarrow P_1 \Rightarrow \ldots \Rightarrow P_n$, where
 rule~R7 has been applied to the set $P_{i}$, with $0\!<\!i\!<\!n$,
it is often more convenient to 
check the validity of Condition~(F) with respect to 
$\textit{Definite}(P_n)$, 
instead of $\textit{Definite}(P_0)\cup \mathit{Defs}_i$, as required by the
hypotheses of Theorem~\ref{thm:sat-preserv}. 
Indeed, in the set $P_{n}$,  predicate {\it diff\/}
is defined by a set of clauses whose variables have all integer or boolean 
type, and hence in checking Condition~(F) we need not reason about
predicates defined over ADTs.
The following Proposition~\ref{prop:fun-preserv}
guarantees that 
Theorem~\ref{thm:sat-preserv} holds even if we 
check Condition (F) with respect to $\textit{Definite}(P_n)$, instead of 
$\textit{Definite}(P_0)\cup \mathit{Defs}_i$.

\begin{proposition}[Preservation of Functionality]
\label{prop:fun-preserv} 
Let $P_0 \Rightarrow P_1 \Rightarrow \ldots \Rightarrow P_n$ be a transformation sequence
using rules~{\rm R1--R7}.
Suppose that Condition~{\rm (U)} of Theorem~$\ref{thm:unsat-preserv}$ holds.
For $i\!=\!0,\ldots,n$,
if an atom $A(X,Y)$ is functional from $X$ to $Y$ with respect to~$\textit{Definite}(P_n)$
and the predicate symbol of $A(X,Y)$ occurs in $\textit{Definite}(P_0)\cup \mathit{Defs}_i$, 
then $A(X,Y)$ is functional from~$X$ to~$Y$ with respect to~$\textit{Definite}(P_0)\cup \mathit{Defs}_i$.
\end{proposition}

\begin{proof}
Let us suppose that $A(X,Y)$ is functional from $X$ to $Y$ 
with respect to $\textit{Definite}(P_n)$,
that is,

\smallskip

$M(\textit{Definite}(P_n))\models \forall\,(A(X,Y) \wedge A(X,Z) \rightarrow Y\eq Z)$.

\smallskip
\noindent
Then, for all (tuples of) ground terms $u$, $v$, and $w$, with $v\!\neq\! w,$ 

\smallskip

$\{A(u,v),A(u,w)\} \not\subseteq M(\textit{Definite}(P_n))$

\smallskip
\noindent
Since $P_0 \Rightarrow P_1 \Rightarrow \ldots \Rightarrow P_n$ is a transformation sequence
using rules R1--R7 such that Condition (U) of Theorem~\ref{thm:unsat-preserv} holds, 
then by Theorem~\ref{thm:cons},

\smallskip

$\{A(u,v),A(u,w)\} \not\subseteq M(\textit{Definite}(P_0) \cup \mathit{Defs}_n)$

\smallskip
\noindent
Hence, for $i\!=\!0,\ldots,n$,

\smallskip

$\{A(u,v),A(u,w)\} \not\subseteq M(\textit{Definite}(P_0) \cup \mathit{Defs}_i)$

\smallskip
\noindent
Thus, 

\smallskip

$M(\textit{Definite}(P_0) \cup \mathit{Defs}_i)\models \forall\,(A(X,Y) \wedge A(X,Z) \rightarrow Y\eq Z)$.
\hfill $\Box$
\end{proof}

\section{A Method for Checking the Satisfiability of CHCs through ADT Removal}
\label{subsec:solvingCHCs}

In this section we put together the results presented in Sections~\ref{sec:TransfRules},
\ref{sec:Strategy}, and~\ref{sec:Completeness} and we define a method 
for checking whether or not a set $P_0$ of CHCs is satisfiable. We 
proceed as follows: (i)~first, we
construct a transformation sequence
$P_0 \Rightarrow P_1 \Rightarrow \ldots \Rightarrow P_n$ using Algorithm~\Diff,
and (ii)~then we apply a CHC solver to $P_n$.
If the solver is able to prove the satisfiability of~$P_n$, 
then, by Theorem~\ref{thm:soundness-AlgorithmR}, $P_0$ is satisfiable.
If 
the solver proves the unsatisfiability of $P_n$ and 
Conditions~(E) and~(F) are both fulfilled during the execution of~\Diff, 
then, by Theorem~\ref{thm:completeness-AlgorithmR}, $P_0$ is unsatisfiable.

Now, (i)~Condition~(E) can be checked by simply
inspecting the substitution computed when applying the Folding Rule~R3
during the {\it Diff-Define-Fold} procedure.
(ii)~Condition~(F1), by 
Proposition~\ref{prop:fun-preserv}, 
can be checked by proving,
for every difference predicate \textit{diff} that has been used in applying
the Differential Replacement Rule~R7,
the satisfiability of the following set of clauses

\vspace{1mm}

\noindent
\hspace{6mm} $D_n ~\cup~ \{\textit{false}\leftarrow Y_1\!\neq \!Y_2, 
\ \textit{diff\/}(T,W,Y_1),\ \textit{diff\/}(T,W,Y_2)\}$,

\vspace{1mm}\noindent
where $D_n$ is the set of clauses defining \textit{diff\/} in $\textit{Definite}(P_n)$.
(iii)~Finally, Conditions~(F2) and (F3)
can be checked by inspecting, for every application of
rule~R7, the clauses $C$ and $\widehat D$ involved in that application
(see Definition~\ref{def:condF}).

\smallskip
{In previous sections we have seen in action our method for proving the 
satisfiability of a set of CHCs. In the following example we show an application of our method to prove unsatisfiability of a set of CHCs.}

\begin{example}\label{ex:func}
\newcommand{\oneast}{{\tts 1$^{\textstyle *}$}}
\newcommand{\oneastf}{{\tts 10$^{\textstyle *}$}}
Let us consider again our introductory example of 
Section~\ref{sec:IntroExample} where we started 
from the initial set
{\it RevCls} made out of clauses~{\tts 1}--{\tts 9}.
Let us suppose that we want to 
check the satisfiability of the set \textit{RevCls}$^{*}$ of clauses
that includes clauses~{\tts 2}--{\tts 9}
and the following clause \oneast, instead of clause~{\tts 1}: 

\vspace{1mm}
\noindent
\oneast. ~{\tts false :- N2=\textbackslash=N0-N1, append(Xs,Ys,Zs), reverse(Zs,Rs), 

\hspace{16mm}len(Xs,N0), len(Ys,N1), len(Rs,N2).}

\vspace{1mm}

\noindent
Clause~\oneast~differs from clause~{\tts 1} because of the constraint 
`{\tts N2=\textbackslash=N0-N1}', instead of `{\tts N2=\textbackslash=N0+N1}'.
The set \textit{RevCls}$^{*}$ %
is unsatisfiable because 
the body of clause~{\tts \oneast} holds, in particular, 
for {\tts Xs=[]} and {\tts Ys=[Y]},
where {\tts Y} is any integer.

Algorithm~\Diff~works for the input set $\textit{RevCls}^{*} =$ \{{\tts \oneast\!\!,2,...,9}\}
exactly as described in Section~\ref{sec:Strategy} for the set
$\textit{RevCls} = \{\mbox{\tts 1,2,...,9}\}$,
except that clause {\tts \oneast}, instead of clause {\tts 1}, is folded
by using the definition:

\vspace{-2mm}

{\small
\begin{verbatim}
D1. new1(N0,N1,N2) :- append(Xs,Ys,Zs), reverse(Zs,Rs), len(Xs,N0), 
                      len(Ys,N1), len(Rs,N2).\end{verbatim}
}
\vspace{-2mm}

\noindent
thereby deriving the following clause, instead of clause {\tts 10}:

\vspace{1mm}
\noindent
\oneastf. {\tts  false :- N2=\textbackslash=N0-N1, new1(N0,N1,N2).}

\vspace{1mm}
\noindent
Thus, the output of Algorithm~\Diff~is \textit{TransfRevCls}$^{*}$ = 
\{{\tts \oneastf\!\!,}{\tts 15,17,18,19,20,21}\}
(for clauses {\tts 15,17,18,19,20,21} see Section~\ref{sec:IntroExample}).
The CHC solver Eldarica proves that \textit{TransfRevCls}$^{*}$ is an unsatisfiable set of clauses.

Now, in order to conclude that also the input set \textit{RevCls}$^{*}$ is unsatisfiable,
we apply Theorem~\ref{thm:completeness-AlgorithmR} and Proposition~\ref{prop:fun-preserv}.
We look at the transformation sequence constructed by Algorithm~\Diff~and we 
check that both Conditions~(E) and (F) are fulfilled.

Condition~(E) is fulfilled because each time we apply the folding rule,
the substitution $\vartheta$ is the identity (see, in particular, the folding
step for deriving clause~\oneastf\ above, and also the folding steps in Example~\ref{ex:rev-continued2}).

Now, let us check Condition~(F). %
We have that clauses~$C$ and~$\widehat{D}$ occurring in 
Definition~\ref{def:condF} are clauses~{\tts 14} 
and~{\tts D3},
respectively.  For the reader's convenience, we list them here:

\vspace{-2mm}

{\small
\begin{verbatim}
14. new1(N01,N1,N21) :- N01=N0+1, append(Xs,Ys,Zs), reverse(Zs,Rs), 
                        len(Xs,N0), len(Ys,N1), snoc(Rs,X,R1s), 
                        len(R1s,N21).
D3. diff(X,N2,N21) :- snoc(Rs,X,R1s), len(R1s,N21), len(Rs,N2).
\end{verbatim}
}

\vspace{-2mm}

\noindent
With reference to Definition~\ref{def:condF} (see Example~\ref{ex:rev-continued2}), we have:

\smallskip
\noindent
$F(X;Y)$\! =\! ({\tts  snoc(Rs,X,R1s)\!,\;len(R1s,N21)}), \ \ 
$X$={\tts  (Rs,X)}, \ \ $Y$={\tts  (R1s,N21)},

\noindent
$R(V;W)$\! =\! {\tts len(Rs,N2)}, \ \ $V$=\,{\tts (Rs)}, \ \ $W$=\,{\tts (N2)}.

\smallskip

\noindent 
Condition~(F1) requires that atom 
{\tts diff(X,N2,N21)}
be functional from {\tts (X,N2)} to~{\tts N21}
 with respect to
${\mathit{Definite}}(\mathit{RevCls}^{*})$.
In order to check this functionality, by Proposition~\ref{prop:fun-preserv},
it suffices to check the satisfiability of the set consisting of following clause:

\vspace{1mm}
\noindent
{\tts 22. false :- N21=\textbackslash=N22, diff(X,N2,N21), diff(X,N2,N22).}
\vspace{1mm}

\noindent
together with clauses {\tts 20} and {\tts 21}, which define {\tts diff} in 
${\mathit{Definite}}(\mathit{TransfRevCls}^{*})$. We recall them here:

\vspace{1mm}
\noindent
{\tts 20. diff(X,N0,N1) :- N0=0, N1=1.}

\noindent	
{\tts 21. diff(X,N0,N1) :- N0=N+1, N1=M+1, diff(X,N,M).}

\vspace{1mm}

\noindent
The CHC solver Eldarica proves the satisfiability of the set \{{\tts 20,21,22}\} 
of clauses by computing the following model:

\vspace{1mm}

{\tts diff(X,N2,N21) :- N21=N2+1, N2>=0.}

\vspace*{1mm}

\noindent
Thus, Condition (F1) is fulfilled. 
Condition~(F2) requires that 
\{{\tts R1s,N21}\} $\cap$ \{{\tts Rs}\}$\,=\! \emptyset$, and
Condition~(F3) requires that \{{\tts R1s}\} $\cap$ \{{\tts Xs,Ys,Zs,Rs}\}$\,=\! \emptyset$.
They are both fulfilled. 
Therefore, by Theorem~\ref{thm:completeness-AlgorithmR}, we conclude, as desired,
that the input
set \textit{RevCls}$^{*}\!$ of clauses is unsatisfiable.\hfill $\Box$

\vspace*{-3mm}

\noindent
\end{example}

%
%
%
%
%
%
%
%
%
%
%
%
%
%
%
%
%
%
%


\section{Experimental Evaluation}
\label{sec:Experiments}
\newcommand{\downsp}{$_{_{_{_{~}}}}$}
\newcommand{\upsp}{\rule{0mm}{3.5mm}}
\newcommand{\updownsp}{\rule{0mm}{3.5mm}$_{_{_{_{~}}}}$}

In this section we present the experimental evaluation we have performed
for assessing the effectiveness of our transformation-based CHC %
{satisfiability checking} method.

We have implemented our method %
in a tool called {\sc AdtRem}
and we have compared the results obtained by running our tool with those
obtained by running:
(i) the CVC4 SMT solver~\cite{ReK15} extended with inductive reasoning and lemma
generation,
(ii) the AdtInd solver~\cite{Ya&19}, which makes use of a syntax-guided
synthesis strategy~\cite{Al&13} for lemma generation, and
(iii) the Eldarica CHC solver~\cite{HoR18}, which combines
predicate abstraction~\cite{GrS97} with counterexample-guided
abstraction refinement~\cite{Cl&03}.
{\sc AdtRem} is available at {\small{\url{https://fmlab.unich.it/adtrem/}}}.

\subsection{The workflow of the {\sc AdtRem} tool}
\label{subsec:AdtRem-workflow}

Our {\sc AdtRem} tool implements the satisfiability checking method %
presented in Section~\ref{subsec:solvingCHCs} as follows.
First, {\sc AdtRem}
 makes use of the VeriMAP system~\cite{De&14b} to perform the steps
specified by Algorithm~\Diff. It takes as input a
set~$P_0$ of CHCs and, if it terminates, it
 produces as output a set~$P_n$ of CHCs that 
have basic types.
Then, in order to show the satisfiability of $P_{0}$, 
{\sc AdtRem} invokes the Eldarica CHC solver to
show the satisfiability of~$P_n$.

If Eldarica proves that $P_n$ is satisfiable,
then, by the soundness of
the transformation Algoritm~\Diff~(see Theorem~\ref{thm:soundness-AlgorithmR}),
$P_0$ is satisfiable and
{\sc AdtRem} returns the answer `\textit{sat}'.
In particular, the implementation of the \textit{Unfold} procedure
enforces Condition~(U) of
Theorem~\ref{thm:unsat-preserv}, which indeed ensures the 
soundness of Algoritm~\Diff.

If Eldarica proves that $P_n$ is unsatisfiable 
by constructing a counterexample CEX,
{\sc AdtRem} proceeds by checking whether or not Conditions (E)
and (F), which guarantee the completeness of Algoritm~\Diff, hold (see Theorem~\ref{thm:completeness-AlgorithmR}).
In particular, during the execution of Algoritm~\Diff,
{\sc AdtRem} checks Condition (E) when applying the
Folding Rule R3, and Conditions (F2) and (F3) when applying the Differential Replacement Rule~R7.
{\sc AdtRem} marks all clauses in $P_n$ derived by a sequence of transformation 
steps
where one of the Conditions (E), (F2), and (F3) is not satisfied.

Then, {\sc AdtRem} looks at the counterexample CEX constructed by Eldarica
to verify whether or not any instance of the marked clauses is used
for the construction of CEX.
If this is the case, it is not possible to establish the unsatisfiability of
$P_0$ from the unsatisfiability of~$P_n$ (as completeness of 
Algoritm~\Diff~may not hold) and hence {\sc AdtRem} returns
the answer `\textit{unknown}' as the result of the satisfiability 
check for $P_{0}$.

Otherwise, if {no instance of a marked clause}
is used in CEX, then {\sc AdtRem} proceeds by checking
Condition~(F1) of Definition~\ref{def:condF}. 
{Recall that, by Proposition~\ref{prop:fun-preserv}, Condition~(F1) can be checked 
by inspecting the clauses defining the differential 
predicates in $\textit{Definite}(P_n)$.}
To perform this check,
for each differential predicate $\textit{diff}_k$ introduced by 
Algorithm~\Diff~and occurring in CEX, {\sc AdtRem}
produces {a clause, call it $\textit{fun-diff}_k$,} of the
form:

\smallskip

$~\textit{false}\leftarrow$ $O_1\!\neq\!O_2,$ 
$\textit{diff}_k(I,O_1), \textit{diff}_k(I,O_2)$

\smallskip

\noindent
where $I$ and $O_i$, for $i=1,2,$ are the tuples of input and output
variables, respectively, of $\textit{diff}_k(I,O_i)$.
Then, {\sc AdtRem} runs Eldarica to check the satisfiability
of $\bigcup_k\{\textit{fun-diff}_k\}\cup D_n$,
where $D_n$ is the set of clauses defining $\textit{diff}_k$ in
$\textit{Definite}(P_n)$ (see Example~\ref{ex:func}).
Now, this satisfiability check may succeed (Case 1) or may not succeed (Case 2).
In Case 1, 
Condition~(F1) holds and, by the completeness of 
Algorithm~\Diff~(see Theorem~\ref{thm:completeness-AlgorithmR}), 
we have that $P_0$ is %
unsatisfiable and {\sc AdtRem} returns the answer `\textit{unsat}'. 
In Case 2, Condition~(F1) cannot be shown,
and {\sc AdtRem} returns the answer `\textit{unknown}' as the result of the
satisfiability of~$P_{0}$.
\subsection{Benchmark suite}
Our benchmark suite consists of 251 {verification} %
problems over inductively
defined data structures, such as lists, queues, heaps, and trees.
Out of these 251 problems, 168 of them
 refer to properties that hold (\textit{valid}
properties) and the remaining 83 refer to properties that do not hold
(\textit{invalid} properties).

The 168 problems specifying valid properties have been adapted from the
benchmark suite considered by Reynolds and Kuncak~\cite{ReK15}, and
originate from benchmarks used by various theorem provers, such as
CLAM~\cite{IrB96}, HipSpec~\cite{Cl&13}, IsaPlanner~\cite{DiF03,Jo&10},
and Leon~\cite{Su&11}.
In particular, we have considered Reynolds and Kuncak's `dtt'
encoding where natural numbers are represented using the built-in
SMT type {\it Int}.
From those problems we have discarded:
(i)~the ones that do not use ADTs, and
(ii)~the ones that cannot be directly represented in Horn clause format.
In order to make a comparison between our approach and Reynolds-Kuncak's
one on a level playing field, since {\sc AdtRem} supports neither higher
order functions nor user-provided lemmas,
(i)~we have replaced higher order functions by suitable first order instances,
and
(ii)~we have removed all auxiliary lemmas from the formalization of the  
problems.
We have also used LIA constraints, instead of the basic functions recursively
defined over natural numbers, such as the {\it plus} function and
{\it less-or-equal} relation, so that the solver can deal with them by using
the LIA theory.

The 83 problems specifying invalid properties have been obtained from those
specifying valid properties by either negating the properties or
modifying {the definitions} of 
the predicates on which the properties depend.

The benchmark suite is available at {\small{\url{https://fmlab.unich.it/adtrem/}}}.

\subsection{Experiments}

We have performed the following experiments.

\smallskip

\noindent\hangindent=4mm
1. We have run the `cvc4+ig' configuration of the CVC4 solver extended
with inductive reasoning and also the AdtInd solver on the 251 
{verification} %
problems in SMT-LIB format.

\noindent\hangindent=4mm
2. Then, we have translated each {verification} %
problem into
a set, call it $P_0$, of CHCs in the Prolog-like syntax supported by {\sc AdtRem} by using a
modified version of the SMT-LIB parser of the ProB system~\cite{LeB03}.
We have run Eldarica {v2.0.5}, which uses no induction-based
mechanism for handling ADTs, to check the satisfiability of the
SMT-LIB translation of $P_0$\footnote{We have also performed analogous experiments on the set of valid
properties by using Z3-SPACER~\cite{Ko&13}, instead of Eldarica,
as reported on an earlier version of this paper~\cite{De&20a}.
The results of those experiments, which we do not report here,
are very similar to those shown in	Table~\ref{tab:Exper-Results}.}\!.

\noindent\hangindent=4mm
3. Finally, we have run {\sc AdtRem} to check the satisfiability %
of $P_0$.
If {\sc AdtRem} returns `\textit{sat}', the property specified by $P_0$ 
is reported to be valid.
If {\sc AdtRem} returns `\textit{unsat}', the property is reported to be invalid.

\smallskip
\noindent
Experiments have been performed 
on an Intel Xeon CPU E5-2640 2.00GHz with 64GB RAM under CentOS
and for each problem we have set a timeout limit of 300 seconds.

\subsection{Evaluation of Results}

The results of our experiments 
are summarized in the following four tables.

\noindent\hangindent=3mm 
-- 
In Table~\ref{tab:Exper-Results} ({\it Solved problems})
we report the number of problems solved by each tool, also %
classified by the type of property (valid or invalid).
Columns 3--6 are labeled by the name of the tool
and the last column reports the results of the `Virtual Best' tool,
that is, the number of problems solved by at least one of the  
tools we have considered.

\noindent\hangindent=3mm 
--
In Table~\ref{tab:Exper-Results-unique} ({\it Uniquely solved problems}) we report,  for each tool, the number of
uniquely solved problems, that is, the number of problems solved by that tool and not solved 
by any of the other tools. By definition, there are no  problems uniquely solved
by the `Virtual Best' tool. %

\noindent\hangindent=3mm 
--
In order to assess the difficulty of the benchmark problems,
we have computed, for each problem, the number of tools that are able to solve it.
The results are reported in Table~\ref{tab:Exper-Results-howmany} 
({\it Benchmark difficulty}). 
\noindent\hangindent=3mm 
--
In Table~\ref{tab:Exper-Results-adtrem} ({\it  Termination 
of Algorithm~\Diff~and effectiveness of {\sc AdtRem}}) we report the number of problems
for which the ADT removal algorithm \Diff~terminates 
and the percentage of those problems solved by {\sc AdtRem}.
\vspace*{-4mm}
\begin{table}[!ht]
\begin{center}
\begin{tabular}{|@{\hspace{-2mm}}l@{\hspace{-3mm}}|r@{\hspace{1mm}}||r@{\hspace{2mm}}|r@{\hspace{2mm}}|r@{\hspace{2mm}}|r@{\hspace{2mm}}||r@{\hspace{1mm}}|}
	\hline
       {\parbox[center]{22mm}{\begin{center}Type of\\ properties\end{center}}} &
       {\parbox[top]{17mm}{\begin{center}Number of\\ problems\end{center}}} &
    {\parbox[top]{13mm}{\vspace*{-2mm}\center CVC4\!\\[-.5mm]
	 with\!\\[-.5mm] induction\!\\[2mm]}}			&
	  {~AdtInd} 								&
	   {~Eldarica} 								&
	   {{ {\textsc{ AdtRem}}} } & 
	  {~Virtual Best}
	  	   \\ \hline \hline
			{~~~~Valid} & 168 & 75 & 62 & 12 &  115 & 127\,\\
			{~~~~Invalid} &  83 & 0 &  1   & 75 &  61  & 83\,\\
			\hline
			{~~~~Total} & 251 & 75 &  63 & 87 & 176 & 210\,\\
			\hline
\end{tabular}
\vspace{2mm}
		\caption{{\small {\it Solved problems.}
		Number of problems solved by each tool.
		}}\label{tab:Exper-Results}
\end{center}
\vspace*{-8mm}
\end{table}
\vspace*{-2mm}

Table~\ref{tab:Exper-Results} shows that, on our benchmark,
{\sc AdtRem} compares favorably to all other tools we have considered.

On problems with valid properties, {\sc AdtRem} performs better
than solvers extended with inductive reasoning, such as CVC4 
and AdtInd.
{\sc AdtRem} performs better than those two tools also on problems
with invalid properties on which CVC4 and AdtInd show %
poor results.
This poor outcome may be due to the fact that those tools were designed
with the aim of proving theorems rather than finding counterexamples 
to non-theorems.

Table~\ref{tab:Exper-Results} also shows that the ADT removal performed by 
Algorithm~\Diff, implemented by {\sc AdtRem}, considerably increases
the  overall
effectiveness of the CHC solver Eldarica,
without the need for any inductive reasoning support.
In particular, Eldarica is able to solve {87} problems
out of  {251} {\em before} the application of Algorithm~\Diff~(see 
Column `Eldarica'),
while {\sc AdtRem} solves {176} problems by using Eldarica {\em after} the application of 
Algorithm~\Diff~(see 
Column `{\sc AdtRem}').

The gain in effectiveness is very high on problems with valid properties,
where Eldarica solves
{12} problems out of {168},
while {\sc AdtRem} solves {115} problems by applying Eldarica after the removal of ADTs.

On problems with invalid properties
Eldarica is already
very effective before the removal of ADTs
and is able to solve  {75} %
problems out of  {83},
whereas the number of problems solved by
{\sc AdtRem}
is only {61},
which are all the problems with invalid properties for which
 Algorithm~\Diff~terminates
(see Table~\ref{tab:Exper-Results-adtrem}).
Note, however, that by inspecting the detailed results of our experiments
(see {\small{\url{https://fmlab.unich.it/adtrem/}}}),
we have found 8 problems with invalid properties solved by {\sc AdtRem},
which are not solved by Eldarica before ADT removal.

\smallskip

\begin{table}[!ht]
\vspace*{-4mm}
\begin{center}
\begin{tabular}{|@{\hspace{-2mm}}l@{\hspace{-3mm}}|r@{\hspace{2mm}}||r@{\hspace{2mm}}|r@{\hspace{2mm}}|r@{\hspace{2mm}}|r@{\hspace{2mm}}||r@{\hspace{1mm}}|}
			\hline
    {\parbox[center]{22mm}{\begin{center}Type of\\ properties\end{center}}} &
    {\parbox[top]{16mm}{\begin{center}Number of\!\!\\ problems\!\!\end{center}}} &
	    {\parbox[top]{13mm}{\vspace*{-2mm}\center CVC4\\[-.5mm]
			with\\[-.5mm] induction\!\\[2mm]}}					 &
	  {~AdtInd} 											 &
	   {~Eldarica} 											&
	   {{ {\textsc{ AdtRem}}} } & 
	  {~Virtual Best}
	  	   \\ \hline \hline
			{~~~~Valid} & 50     & 3 &  2 &     1 &  44 & -\,\\
			{~~~~Invalid} &  29 &  0 &  0   &  22 &  7  & -\,\\
			\hline
			{~~~~Total} &  79 &    3 &   2 &    23 &  51 & -\,\\
			\hline
\end{tabular}
\vspace{2mm}
		\caption{{\small {\it Uniquely solved problems.}
		Number of uniquely solved problems by each tool.
		}}\label{tab:Exper-Results-unique}
\end{center}
\vspace*{-8mm}
\end{table}

Table~\ref{tab:Exper-Results-unique} shows further evidence that 
the overall performance of {\sc AdtRem} is higher than that of the other tools.
Indeed, the number of problems solved by {\sc AdtRem} only is larger
than the number of problems uniquely solved by any other tool.
In particular, {\sc AdtRem} uniquely solves 51 problems (44 with valid properties, 
7 with invalid properties) and Eldarica uniquely solves 23 problems (1  
with valid property, 22 with invalid properties).

\begin{table}[!ht]
\vspace*{-4mm}
\begin{center}
\begin{tabular}{|@{\hspace{-2mm}}l@{\hspace{-3mm}}|r@{\hspace{2mm}}||r@{\hspace{2mm}}|r@{\hspace{2mm}}|r@{\hspace{2mm}}|r@{\hspace{2mm}}|r@{\hspace{1mm}}|}
			\hline
       {\parbox[center]{21mm}{\begin{center}Type of\\ properties\end{center}}} &
       {\parbox[top]{16mm}{\begin{center}Number of\\ problems\end{center}}} &
	     {~Unsolved}  & 
	     {~\parbox[top]{13mm}{\begin{center}Uniquely \\ solved\end{center}}}    & 
	     {~\parbox[top]{14mm}{\begin{center}Solved by \\ two tools\end{center}}}    & 
	     {~\parbox[top]{15mm}{\begin{center}Solved by \\ three tools\end{center}}}    & 
	     {~\parbox[top]{14mm}{\begin{center}Solved by \\ all tools\end{center}}}     
	  	   \\ \hline \hline
			{~~~~Valid} & 168  & 41& 50&25&44&8\,\\
			{~~~~Invalid} & 83 &  0 & 29&54&0&0\,\\
			\hline
			{~~~~Total} &  251 &  41&79&79&44 & 8\,\\
			\hline
\end{tabular}
\vspace{2mm}
		\caption{{\small {\it Benchmark difficulty.}
		Number of problems grouped by the number of tools that were able to solve them.
		}}\label{tab:Exper-Results-howmany}

\end{center}
\vspace*{-8mm}
\end{table}

Table~\ref{tab:Exper-Results-howmany} illustrates the degree of
difficulty of the problems of the benchmark for the tools we have considered.
Indeed, out of the 251 problems in the benchmark, 41 of them are solved by no tool,
and only 8 problems are solved by all tools.

\begin{table}[!ht]
\begin{center}
\begin{tabular}{|@{\hspace{-2mm}}l@{\hspace{-3mm}}|r@{\hspace{2mm}}||r@{\hspace{2mm}}|r@{\hspace{2mm}}|}
			\hline
       {\parbox[center]{24mm}{\begin{center}Type of\\ properties\end{center}}} &
       {\parbox[top]{18mm}{\begin{center}Number of\\ problems\end{center}}} &
         {\parbox[top]{20mm}{\begin{center}   
             {Algorithm~\Diff \\~terminates}    
	                       \end{center}}}
	                       				 &
	   {\parbox[top]{26mm}{\begin{center}Percentage solved\\ by {\sc AdtRem}    
	                       \end{center}}}

	  	   \\ \hline \hline
			{~~~~Valid} & 168     & 117 &  $98\%$ \,\\
			{~~~~Invalid} &  83 &  61 &       $100\%$ \,\\
			\hline
			{~~~~Total} &  251 &    178 & $99\%$~\,\\
			\hline
\end{tabular}
\vspace{2mm}
		\caption{{\small {\it  Termination 
					of Algorithm~\Diff~and effectiveness of {\sc AdtRem}.}
		Number of problems for which the ADT removal algorithm terminates 
and the percentage of those problems solved by {\sc AdtRem}.
		}}\label{tab:Exper-Results-adtrem}

\end{center}
\vspace*{-8mm}
\end{table}

Table~\ref{tab:Exper-Results-adtrem} shows that Algorithm~\Diff~terminates
quite often and, whenever it terminates
Eldarica is able to check the satisfiability of the derived set of clauses 
in almost all cases.
Indeed, Algorithm~\Diff~terminates on {178} 
problems out of 251, and {\sc AdtRem}
solves 176 problems out of those {178}.

Note that the use of the Differential Replacement Rule R7 
{(which is a
novel rule we have used in this paper)}
has a positive effect on the termination of Algorithm~\Diff.
In order to assess this effect,
we have implemented
a modified version of the ADT removal algorithm~\Diff, 
called~{$\mathcal{R}^{\circ}$}\!, which
\textit{does not} introduce difference predicates.
Indeed, in {$\mathcal{R}^{\circ}$} the  {\it Diff-Introduce} case of the {\it Diff-Define-Fold}
Procedure of Figure~\ref{fig:Diff} 
is never executed.
We have applied Algorithm~{$\mathcal{R}^{\circ}$} to the 168 problems with 
valid properties
and it terminated only on {94} of them, while  
\Diff~terminates on 117 (see Table~\ref{tab:Exper-Results-adtrem}).
Details are given in {\small{\url{https://fmlab.unich.it/adtrem/}}}.

\smallskip

The effectiveness of the solvers 
that use induction we have considered, namely, CVC4 and AdtInd,
may depend on the supply of suitable lemmas
to be used for proving the main conjecture and also on the 
representation of the natural numbers.
Indeed, further experiments we have performed (see {\small{\url{https://fmlab.unich.it/adtrem/}}}) show that,
on problems with valid properties,
CVC4 and AdtInd solve 102 (instead of 75) and 64 (instead of 62) problems, respectively,
when auxiliary lemmas are added as extra axioms.
If, in addition, we consider the
`dti' encoding of the natural numbers~\footnote{In the `dti' encoding, natural numbers are represented using both the  built-in type  {\it Int}
	and the ADT inductive definition with the zero and successor constructors~\cite{ReK15}.}\!,
CVC4 and AdtInd solve 139 and 59 problems, respectively.
Our results show (see Table~\ref{tab:Exper-Results}) that in most cases {\sc AdtRem} needs neither
those extra axioms nor
that sophisticated encoding.

\smallskip

Finally, in Table~\ref{tab:examples} we report some problems with valid properties
solved by {\sc AdtRem}
that are not solved by  CVC4 with induction, nor by AdtInd, nor by  Eldarica.
CVC4 with induction and AdtInd are not able to solve those problems even 
if we take their formalizations with auxiliary lemmas and different encodings 
of the natural numbers.
In Table~\ref{tab:examples2} we report problems with valid properties
solved by CVC4 with induction, or by AdtInd, or by  Eldarica, 
that are not solved  by {\sc AdtRem}.

\begin{table}[!ht]
\vspace{-3mm}
\begin{center}
\begin{tabular}{|@{\hspace{1mm}}l@{\hspace{0mm}}|@{\hspace{2mm}}l@{\hspace{10mm}}|}

\hline

{\it Problem} & {\it Property}  \\ \hline\hline
CLAM goal4\hspace*{1.3cm} \updownsp &  $\forall x.\ \mathit{len (append(x,\!x)) =} ~2 \  \mathit{len(x)}$ \\ \hline
CLAM goal6\hspace*{1.3cm}\updownsp &  $\forall x,y.\ \mathit{len (rev(append(x,\!y))) = len(x) + len(y)}$ \\ \hline
IsaPlanner goal52 \updownsp &  $\forall n,l.\ \mathit{count(n,\!l) =  count(n, rev(l))}$  \\ \hline
IsaPlanner goal80 \updownsp &   $\forall l.\ \mathit{sorted (sort(l))}$   \\ \hline
\hline
\end{tabular}\vspace*{3mm}

\caption{{\small  Problems solved by {\sc AdtRem} and solved by neither {\rm CVC4  with induction nor AdtInd nor Eldarica.}}}
\label{tab:examples}
\end{center}
\vspace*{-6mm}
\end{table}
\vspace*{-8mm}

\begin{table}[!ht] %
\begin{center}

\begin{tabular}{|@{\hspace{1mm}}l@{\hspace{1mm}}|@{\hspace{1mm}}l@{\hspace{1mm}}|@{\hspace{1mm}}c@{\hspace{1mm}}|}
 \hline
{\it Problem} & {\it Property} & {\it Solved by}  \\ \hline\hline

CLAM goal18 &  $\forall x,\!y.\ \mathit{rev(append(rev(x),\! y))\!=\! append(rev(y),\!x)}$ &  {\parbox[top]{18mm}{\vspace*{-2mm}\center CVC4
	with\\[-.5mm] induction\!\\[1.2mm]}}	  \\ \hline
	
CLAM goal76 \upsp\downsp&   $\forall x,\!y.\ \mathit{append(revflat(x),\! y)\!=\! qrevaflat(x,\!y)}$  & AdtInd  \\ \hline \\[-4mm]

{\parbox[top]{22mm}{\upsp Leon amortize-\\queue-goal3\downsp}} &  $\forall x.\ \mathit{len(qrev(x))\!=\! len(x)}$ & Eldarica  \\ \hline

\end{tabular}\vspace*{3mm}

\caption{{\small Problems solved by  {\rm CVC4 with induction or AdtInd or Eldarica} and not solved by {\sc AdtRem}.
}}
\label{tab:examples2}
\end{center}
\vspace*{-6mm}
\end{table}
\vspace*{-5mm}
%

%

\section{Related Work and Conclusions}  
\label{sec:RelConcl}

This paper is an improved, extended version of a paper that appears in the Proceedings of
the 10th International Joint Conference on Automated Reasoning (\mbox{IJCAR}~2020)~\cite{De&20a}.  
The paper was also presented at the 35th Italian Conference 
on Computational Logic (CILC 2020).
Besides detailed proofs and examples, the main, new contribution
consists in addressing the problem of the {\em completeness}
of the transformations. In the IJCAR paper, we proved only a {\em soundness} property, 
ensuring that if the transformed clauses are satisfiable,
then so are the original clauses.
In this paper, we identify some sufficient conditions, 
related to the functionality of
the difference predicates introduced by the transformation algorithm~\Diff,
{which guarantee completeness (that is, the converse of soundness) 
stating that if the transformed clauses are unsatisfiable,
then so are the original clauses.} We have also extended our benchmark
and our implementation, and we have shown that those sufficient conditions 
are indeed satisfied
in many interesting, non trivial examples.
Finally, we have extended the comparison of our experimental results with
the ones obtained by the {AdtInd} solver that extends CHC solving 
with induction on the ADT structure and lemma generation~\cite{Ya&19}.

Inductive reasoning is supported, with different degrees of human
intervention, by many theorem provers, such as
ACL2~\cite{ACL2},
CLAM~\cite{IrB96},
Isabelle~\cite{IsaHOLBook02},
HipSpec~\cite{Cl&13},
Zeno~\cite{So&12}, and
PVS~\cite{Ow&92}.
The combination of inductive reasoning and SMT solving techniques
has been exploited  by many tools for program
verification~\cite{Lei12,PhW16,ReK15,Su&11,Un&17,Ya&19}.

Leino~\cite{Lei12} integrates inductive reasoning into the Dafny program verifier
by implementing a simple strategy that rewrites user-defined
properties that may benefit from induction
into proof obligations to be discharged by Z3~\cite{DeB08}.
The advantage of this technique is that it fully decouples inductive
reasoning from SMT solving.
Hence, no extensions to the SMT solver
are required.

In order to extend CVC4 with induction, Reynolds and Kuncak~\cite{ReK15}
also consider the rewriting of formulas that may take advantage from inductive
reasoning, but this is done dynamically, during the proof search.
This approach allows CVC4 to perform the rewritings lazily,
whenever new formulas are generated during the proof search,
and to use the partially solved
conjecture {for generating}
lemmas that may help in the proof of the initial conjecture.

The issue of generating suitable lemmas during inductive proofs has been also
addressed by Yang et al.~\cite{Ya&19} and implemented in~{AdtInd}.
In order to conjecture new lemmas, their algorithm makes use of a
syntax-guided synthesis strategy driven by a grammar, which is
dynamically generated from user-provided templates and the function
and predicate symbols encountered during the proof search.
The derived lemma~conjectures are then checked by the  SMT solver Z3.
In our approach, the introduction of
difference predicates can be viewed
as the transformational counterpart of lemma generation. 
When we prove the satisfiability of the transformed CHCs, we 
also compute models of the difference predicates, which indeed correspond to 
valid properties. These properties could also be added to the background theory and
used as axioms by solvers or theorem provers
(see the example in Section~\ref{sec:IntroExample}).
The experimental evaluation of Section~\ref{sec:Experiments},
shows that, on a small but non-trivial benchmark, our tool {\sc AdtRem} performs 
better than {AdtInd}.

In order to take full advantage of the efficiency of SMT solvers
in checking satisfiability of quantifier-free formulas over LIA,
ADTs, and finite sets, the Leon verification system~\cite{Su&11} implements an
SMT-based solving algorithm to check the satisfiability of formulas
involving recursively defined first-order functions.
The algorithm interleaves the unrolling of recursive functions
and the SMT solving of the formulas generated by the unrolling.
Leon can be used to prove properties of Scala programs with ADTs and
integrates with the Scala compiler and the SMT solver Z3.
A refined version of that algorithm, restricted to
{\it catamorphisms}, %
has been implemented into a solver-agnostic tool, called RADA~\cite{PhW16}.

In the context of CHCs, Unno et al.~\cite{Un&17} have
proposed a proof system that combines inductive theorem proving with
SMT solving. This approach  %
uses \mbox{Z3} with the PDR engine~\cite{HoB12}
 to discharge proof obligations generated by the proof system,
and has been applied %
to prove relational properties %
of OCaml  programs.

Recent work by Kostyukov et al.~\cite{KostyukovMF21}
proposes a method for proving the satisfiability of CHCs over ADTs
by computing 
models represented by finite tree automata. 
The tool based on this approach, called RegInv, is applied
to the problem of computing invariants of programs that manipulate ADTs
and it is shown to be more practical, in some cases, than
state-of-the-art CHC solvers that compute invariants
represented by first-order logic formulas.
{In our approach we do not provide an explicit representation
of the model of the initial, non-transformed CHCs, 
while the transformed clauses have basic types only,
and thus we need not extend the usual notion of a model.}

The distinctive feature of the
technique presented in this paper is that it does not make
use of any explicit inductive reasoning, but it follows a transformational
approach.
First, the problem of verifying the validity of a universally
quantified formula over ADTs is reduced to the problem of checking the
satisfiability of a set of CHCs.
Then, this set of CHCs is transformed with the aim of
deriving a set of CHCs  {over basic types (such as integers and booleans)
only}, whose satisfiability
implies the satisfiability of the original set.
In this way, the reasoning on ADTs is separated from
the reasoning on satisfiability, which can be performed
by specialized engines for CHCs on basic types
(e.g., Eldarica~\cite{HoR18} and Z3-SPACER~\cite{Ko&13}).
Some of the ideas presented here have been explored in
{previous work}~\cite{De&19b,De&19c},  %
but there
neither formal results nor an automated strategy were presented.

A key success factor of our technique is the introduction of
difference predicates, which, as already mentioned, can be viewed
as a form of {automatic} lemma generation.
Indeed, as shown in Section~\ref{sec:Experiments}, the use of difference predicates
greatly increases the power of CHC solving with respect to previous techniques based on
the transformational approach, which do not use difference predicates~\cite{De&18a}. 

As future work, we plan to apply our transformation-based verification technique
to more complex program properties, such as relational properties~\cite{De&16c,De&17c}.
Another important problem to study is the termination of the ADT removal algorithm~\Diff.
As already mentioned, due to the undecidability of the satisfiability problem for CHCs, there is no sound
and complete algorithm that always terminates and removes all ADTs.
Thus, we can tackle this problem in two ways: (1)~by identifying restricted classes of CHCs
for which a sound and complete ADT removal algorithm terminates, or 
(2)~by retaining soundness only and designing a suitable generalization strategy that guarantees 
the introduction of a finite set of new definitions during ADT removal. 
For Point~(2), it would be interesting to explore 
how to combine the introduction
of difference predicates with the most specific generalization techniques
used in automated theorem proving~\cite{Bun01} and logic program transformation~\cite{De&99,SoG95}.

\section*{Acknowledgements}
\label{sec:Acknowledgements}
We warmly thank Francesco Calimeri, Simona Perri, and Ester Zumpano 
for inviting us to submit this improved, extended version of the paper 
we presented at the 35th Italian Conference 
on Computational Logic (CILC 2020) held in Rende, Italy, 13--15 October 2020.
We also thank the anonymous reviewers 
for their helpful comments and suggestions.


\begin{thebibliography}{10}

\bibitem{Al&13}
R.~Alur, R.~Bodik, G.~Juniwal, M.~M.~K. Martin, M.~Raghothaman, S.~A. Seshia,
  R.~Singh, A.~Solar-Lezama, E.~Torlak, and A.~Udupa.
\newblock Syntax-guided synthesis.
\newblock In {\em 2013 Formal Methods in Computer-Aided Design}, pages 1--8,
  2013.

\bibitem{Apt90}
K.~R. Apt.
\newblock Introduction to logic pro\-gramming.
\newblock In J.~van Leeuwen, editor, {\em Handbook of Theoretical Computer
  Science}, pages 493--576. Elsevier, 1990.

\bibitem{CVC4}
C.~Barrett, C.~L. Conway, M.~Deters, L.~Hadarean, D.~Jovanovic, T.~King,
  A.~Reynolds, and C.~Tinelli.
\newblock {CVC4}.
\newblock In G.~Gopalakrishnan and S.~Qadeer, editors, {\em 23rd {CAV}~'11},
  Lecture Notes in Computer Science 6806, pages 171--177. Springer, 2011.

\bibitem{BaT18}
C.~W. Barrett and C.~Tinelli.
\newblock Satisfiability modulo theories.
\newblock In E.~M. Clarke, T.~A. Henzinger, H.~Veith, and R.~Bloem, editors,
  {\em Handbook of Model Checking}, pages 305--343. Springer, 2018.

\bibitem{Bj&15}
N.~Bj{\o}rner, A.~Gurfinkel, K.~L. McMillan, and A.~Rybalchenko.
\newblock Horn clause solvers for program verification.
\newblock In L.~D. Beklemishev, A.~Blass, N.~Dershowitz, B.~Finkbeiner, and
  W.~Schulte, editors, {\em Fields of Logic and Computation {\rm{(II)}}},
  Lecture Notes in Computer Science 9300, pages 24--51. Springer, 2015.

\bibitem{Bun01}
A.~Bundy.
\newblock The automation of proof by mathematical induction.
\newblock In A.~Robinson and A.~Voronkov, editors, {\em Handbook of Automated
  Reasoning {\rm{(I)}}}, pages 845--911. North Holland, 2001.

\bibitem{MaS13}
A.~Cimatti, A.~Griggio, B.~Schaafsma, and R.~Sebastiani.
\newblock The {MathSAT5 SMT} solver.
\newblock In N.~Piterman and S.~Smolka, editors, {\em 19th TACAS~'13}, Lecture
  Notes in Computer Science 7795, pages 93--107. Springer, 2013.

\bibitem{Cl&13}
K.~Claessen, M.~Johansson, D.~Ros{\'e}n, and N.~Smallbone.
\newblock Automating inductive proofs using theory exploration.
\newblock In M.~P. Bonacina, editor, {\em CADE-24}, Lecture Notes in Artificial
  Intelligence 7898, pages 392--406. Springer, 2013.

\bibitem{Cl&03}
E.~Clarke, O.~Grumberg, S.~Jha, Y.~Lu, and H.~Veith.
\newblock Counterexample-guided abstraction refinement for symbolic model
  checking.
\newblock {\em Journal of the {ACM}}, 50(5):752--794, Sept. 2003.

\bibitem{CoH78}
P.~Cousot and N.~Halbwachs.
\newblock Automatic discovery of linear restraints among variables of a
  program.
\newblock In {\em 5th POPL~'78}, pages 84--96. {ACM}, 1978.

\bibitem{De&14b}
E.~{De~Angelis}, F.~Fioravanti, A.~Pettorossi, and M.~Proietti.
\newblock {VeriMAP}: {A} tool for verifying programs through transformations.
\newblock In {\em 20th TACAS~'14}, Lecture Notes in Computer Science 8413,
  pages 568--574. Springer, 2014.

\bibitem{De&16c}
E.~{De Angelis}, F.~Fioravanti, A.~Pettorossi, and M.~Proietti.
\newblock Relational verification through {H}orn clause transformation.
\newblock In X.~Rival, editor, {\em 23rd {SAS}~'16}, Lecture Notes in Computer
  Science 9837, pages 147--169. Springer, 2016.

\bibitem{De&17b}
E.~{De Angelis}, F.~Fioravanti, A.~Pettorossi, and M.~Proietti.
\newblock Semantics-based generation of verification conditions via program
  specialization.
\newblock {\em Science of Computer Programming}, 147:78--108, 2017.
\newblock Selected and Extended papers from the International Symposium on
  Principles and Practice of Declarative Programming 2015.

\bibitem{De&17c}
E.~{De Angelis}, F.~Fioravanti, A.~Pettorossi, and M.~Proietti.
\newblock Predicate {P}airing for program verification.
\newblock {\em Theory and Practice of Logic Programming}, 18(2):126--166, 2018.

\bibitem{De&18a}
E.~{De Angelis}, F.~Fioravanti, A.~Pettorossi, and M.~Proietti.
\newblock Solving {H}orn clauses on inductive data types without induction.
\newblock {\em Theory and Practice of Logic Programming}, 18(3-4):452--469,
  2018.

\bibitem{De&19b}
E.~{De Angelis}, F.~Fioravanti, A.~Pettorossi, and M.~Proietti.
\newblock Lemma generation for {H}orn clause satisfiability: {A} preliminary
  study.
\newblock In A.~Lisitsa and A.~P. Nemytykh, editors, {\em 7th Int. Workshop on
  Verification and Program Transformation, VPT@Programming}, {EPTCS no.\,299},
  pages 4--18, 2019.

\bibitem{De&19c}
E.~{De Angelis}, F.~Fioravanti, A.~Pettorossi, and M.~Proietti.
\newblock Proving properties of sorting programs: {A} case study in {H}orn
  clause verification.
\newblock In E.~{De Angelis}, G.~Fedyukovich, N.~Tzevelekos, and M.~Ulbrich,
  editors, {\em 6th {HCVS} and 3rd {PERR}}, {EPTCS no.\,296}, pages 48--75,
  2019.

\bibitem{De&20a}
E.~{De Angelis}, F.~Fioravanti, A.~Pettorossi, and M.~Proietti.
\newblock Removing algebraic data types from constrained {H}orn clauses using
  difference predicates.
\newblock In N.~Peltier and V.~Sofronie-Stokkermans, editors, {\em Proceedings
  of the International Joint Conference on Automated Reasoning, IJCAR~2020},
  Lecture Notes in Artificial Intelligence 12166, pages 83--102. Springer,
  2020.

\bibitem{DeB08}
L.~M. de~Moura and N.~Bj{\o}rner.
\newblock Z3: {A}n efficient {SMT} solver.
\newblock In {\em 14th {TACAS}~'08}, Lecture Notes in Computer Science 4963,
  pages 337--340. Springer, 2008.

\bibitem{De&99}
D.~{De Schreye}, R.~Gl{\"u}ck, J.~J{\o}rgensen, M.~Leuschel, B.~Martens, and
  M.~H. S{\o}rensen.
\newblock Conjunctive partial deduction: Foundations, control, algorithms, and
  experiments.
\newblock {\em Journal of Logic Programming}, 41(2--3):231--277, 1999.

\bibitem{DiF03}
L.~Dixon and J.~D. Fleuriot.
\newblock {IsaPlanner}: {A} prototype proof planner in {I}sabelle.
\newblock In F.~Baader, editor, {\em CADE-19}, Lecture Notes in Computer
  Science 2741, pages 279--283. Springer, 2003.

\bibitem{End72}
H.~Enderton.
\newblock {\em A Mathematical Introduction to Logic}.
\newblock Academic Press, New York, 1972.

\bibitem{EtG96}
S.~Etalle and M.~Gabbrielli.
\newblock Trans\-form\-ations of {CLP} modules.
\newblock {\em Theoretical Computer Science}, 166:101--146, 1996.

\bibitem{Fi&04a}
F.~Fioravanti, A.~Pettorossi, and M.~Proietti.
\newblock Transformation rules for locally stratified constraint logic
  programs.
\newblock In K.-K. Lau and M.~Bruynooghe, editors, {\em Program Development in
  Computational Logic}, Lecture Notes in Computer Science 3049, pages 292--340.
  Springer-Verlag, 2004.

\bibitem{Fi&13a}
F.~Fioravanti, A.~Pettorossi, M.~Proietti, and V.~Senni.
\newblock Generalization strategies for the verification of infinite state
  systems.
\newblock {\em Theory and Practice of Logic Programming}, 13(2):175--199, 2013.

\bibitem{GrS97}
S.~Graf and H.~Sa\"idi.
\newblock Construction of abstract state graphs with {PVS}.
\newblock In O.~Grumberg, editor, {\em CAV~'97}, LNCS 1254, pages 72--83,
  Berlin, Heidelberg, 1997. Springer.

\bibitem{Gr&12}
S.~Grebenshchikov, N.~P. Lopes, C.~Popeea, and A.~Rybalchenko.
\newblock Synthesizing software verifiers from proof rules.
\newblock In {\em 33rd ACM SIGPLAN Conf.~Programming Language Design and
  Implementation, PLDI~'12}, pages 405--416, 2012.

\bibitem{HoB12}
K.~Hoder and N.~Bj{\o}rner.
\newblock Generalized property directed reachability.
\newblock In A.~Cimatti and R.~Sebastiani, editors, {\em 15th SAT~'12}, Lecture
  Notes in Computer Science 7317, pages 157--171. Springer, 2012.

\bibitem{HoR18}
H.~Hojjat and P.~R{\"{u}}mmer.
\newblock The {ELDARICA} {H}orn solver.
\newblock In N.~Bj{\o}rner and A.~Gurfinkel, editors, {\em Formal Methods in
  Computer Aided Design, \mbox{FMCAD~2018}}, pages 1--7. {IEEE}, 2018.

\bibitem{IrB96}
A.~Ireland and A.~Bundy.
\newblock Productive use of failure in inductive proof.
\newblock {\em Journal of Automated Reasoning}, 16(1):79--111, Mar. 1996.

\bibitem{JaM94}
J.~Jaffar and M.~Maher.
\newblock Constraint logic programming: {A} survey.
\newblock {\em Journal of Logic Programming}, 19/20:503--581, 1994.

\bibitem{Jo&10}
M.~Johansson, L.~Dixon, and A.~Bundy.
\newblock Case-analysis for rippling and inductive proof.
\newblock In M.~Kaufmann and L.~C. Paulson, editors, {\em Interactive Theorem
  Proving}, Lecture Notes in Computer Science 6172, pages 291--306. Springer,
  2010.

\bibitem{Ka&16}
B.~Kafle, J.~P. Gallagher, and J.~F. Morales.
\newblock {RAHFT}: {A} tool for verifying {H}orn clauses using abstract
  interpretation and finite tree automata.
\newblock In {\em 28th {CAV}~'16, Part I}, Lecture Notes in Computer Science
  9779, pages 261--268. Springer, 2016.

\bibitem{ACL2}
M.~Kaufmann, P.~Manolios, and J.~S. Moore.
\newblock {\em Computer-Aided Reasoning: An Approach}.
\newblock Kluwer Academic Publishers, 2000.

\bibitem{Ko&13}
A.~Komuravelli, A.~Gurfinkel, S.~Chaki, and E.~M. Clarke.
\newblock Automatic abstraction in {SMT}-based unbounded software model
  checking.
\newblock In N.~Sharygina and H.~Veith, editors, {\em 25th {CAV}~'13}, Lecture
  Notes in Computer Science 8044, pages 846--862. Springer, 2013.

\bibitem{KostyukovMF21}
Y.~Kostyukov, D.~Mordvinov, and G.~Fedyukovich.
\newblock Beyond the elementary representations of program invariants over
  algebraic data types.
\newblock In S.~N. Freund and E.Yahav, editors, {\em {PLDI} '21: 42nd {ACM}
  {SIGPLAN} International Conference on Programming Language Design and
  Implementation, Canada, June 20-25, 2021}, pages 451--465. {ACM}, 2021.

\bibitem{Lei12}
K.~Leino.
\newblock Automating induction with an {SMT} solver.
\newblock In V.~Kuncak and A.~Rybalchenko, editors, {\em 13th {VMCAI}}, Lecture
  Notes in Computer Science 7148, pages 315--331. Springer, 2012.

\bibitem{Le&17}
X.~Leroy, D.~Doligez, A.~Frisch, J.~Garrigue, D.~R\'emy, and J.~Vouillon.
\newblock The {OC}aml system, {R}elease 4.06.
\newblock Documentation and user's manual, Institut {N}ational de {R}echerche
  en {I}nformatique et en {A}utomatique, {F}rance, 2017.

\bibitem{LeB03}
M.~Leuschel and M.~Butler.
\newblock {ProB}: A model checker for {B}.
\newblock In {\em FME 2003: Formal Methods}, Lecture Notes in Computer Science
  2805, pages 855--874. Springer, 2003.

\bibitem{Llo87}
J.~W. Lloyd.
\newblock {\em Foundations of Logic Pro\-gramming}.
\newblock Springer-Verlag, Berlin, 1987.
\newblock Second Edition.

\bibitem{IsaHOLBook02}
T.~Nipkow, M.~Wenzel, and L.~C. Paulson.
\newblock {\em Isabelle/HOL: A Proof Assistant for Higher-Order Logic}.
\newblock Springer, 2002.

\bibitem{Ow&92}
S.~Owre, J.~M. Rushby, and N.~Shankar.
\newblock {PVS}: A prototype verification system.
\newblock In D.~Kapur, editor, {\em CADE-11}, pages 748--752. Springer, 1992.

\bibitem{Pe&12a}
A.~Pettorossi, M.~Proietti, and V.~Senni.
\newblock Constraint-based correctness proofs for logic program
  transformations.
\newblock {\em Formal Aspects of Computing}, 24:569--594, 2012.

\bibitem{PhW16}
T.-H. Pham, A.~Gacek, and M.~W. Whalen.
\newblock Reasoning about algebraic data types with abstractions.
\newblock {\em J. Autom. Reason.}, 57(4):281--318, Dec. 2016.

\bibitem{Rab77}
M.~O. Rabin.
\newblock Decidable theories.
\newblock In J.~Barwise, editor, {\em Handbook of Mathematical Logic}, pages
  595--629. North-Holland, 1977.

\bibitem{ReK15}
A.~Reynolds and V.~Kuncak.
\newblock Induction for {SMT} solvers.
\newblock In D.~D'Souza, A.~Lal, and K.~G. Larsen, editors, {\em 16th {VMCAI}},
  Lecture Notes in Computer Science 8931, pages 80--98. Springer, 2015.

\bibitem{Sek09}
H.~Seki.
\newblock On inductive and coinductive proofs via unfold/fold transformations.
\newblock In D.~{De Schreye}, editor, {\em 19th LOPSTR~'09}, Lecture Notes in
  Computer Science 6037, pages 82--96. Springer, 2010.

\bibitem{So&12}
W.~Sonnex, S.~Drossopoulou, and S.~Eisenbach.
\newblock Zeno: {A}n automated prover for properties of recursive data
  structures.
\newblock In C.~Flanagan and B.~K{\"o}nig, editors, {\em 18th TACAS~'12}, pages
  407--421. Springer, 2012.

\bibitem{SoG95}
M.~H. S{\o}rensen and R.~Gl{\"u}ck.
\newblock An algorithm of generalization in positive supercompilation.
\newblock In J.~W. Lloyd, editor, {\em Proceedings of the 1995 International
  Logic Programming Symposium (ILPS '95)}, pages 465--479. MIT Press, 1995.

\bibitem{Su&11}
P.~Suter, A.~S. K\"{o}ksal, and V.~Kuncak.
\newblock Satisfiability modulo recursive programs.
\newblock In E.~Yahav, editor, {\em 18th {SAS}~'11}, Lecture Notes in Computer
  Science 6887, pages 298--315. Springer, 2011.

\bibitem{TaS84}
H.~Tamaki and T.~Sato.
\newblock Unfold/fold trans\-form\-ation of logic pro\-grams.
\newblock In S.-{\AA}. T{{\"a}}rnlund, editor, {\em Proceedings of the Second
  International Conference on Logic Pro\-gramming,~ICLP~'84}, pages 127--138,
  Uppsala, Sweden, 1984. Uppsala University.

\bibitem{TaS86}
H.~Tamaki and T.~Sato.
\newblock A generalized correctness proof of the unfold/fold logic pro\-gram
  trans\-form\-ation.
\newblock Technical Report 86-4, Ibaraki University, Japan, 1986.

\bibitem{Un&17}
H.~Unno, S.~Torii, and H.~Sakamoto.
\newblock Automating induction for solving {H}orn clauses.
\newblock In R.~Majumdar and V.~Kuncak, editors, {\em 29th {CAV}~'17, Part
  {II}}, Lecture Notes in Computer Science 10427, pages 571--591. Springer,
  2017.

\bibitem{Ya&19}
W.~Yang, G.~Fedyukovich, and A.~Gupta.
\newblock Lemma synthesis for automating induction over algebraic data types.
\newblock In T.~Schiex and S.~de~Givry, editors, {\em 25th Int. Conf.
  Principles and Practice of Constraint Programming, {CP~2019}}, Lecture Notes
  in Computer Science 11802, pages 600--617. Springer, 2019.

\end{thebibliography}
\end{document}